\documentclass[final,3p,10pt]{elsarticle}
\usepackage[T1]{fontenc}
\usepackage[utf8]{inputenc}
\usepackage{lmodern}
\usepackage{amsmath,amssymb,amsthm,mathtools}
\usepackage{booktabs,tabularx,array,longtable}
\usepackage{graphicx,xcolor}
\usepackage{microtype}
\usepackage{xurl}
\usepackage{enumitem}
\usepackage{placeins}
\usepackage{hyperref}
\hypersetup{colorlinks=true,linkcolor=blue!45!black,citecolor=blue!45!black,urlcolor=blue!45!black,pdftitle={A Reduced Stochastic Model of Microservice Availability: Algorithms for Identification and Formal Analysis},pdfauthor={Anatoly A. Krasnovsky and Anna Maslovskaya}}
\newtheorem{theorem}{Theorem}
\newtheorem{proposition}{Proposition}
\newcommand{\Prob}{\mathbb{P}}
\newcommand{\E}{\mathbb{E}}
\newcommand{\ind}{\mathbf{1}}
\newcommand{\bits}{\{0,1\}}
\newcommand{\Gstar}{\ensuremath{G^{\star}}}
\newcommand{\GID}{\ensuremath{G_{\mathrm{ID}}}}
\newcommand{\Gzero}{\ensuremath{G_0}}
\newcommand{\Bzero}{\ensuremath{B_0}}

\newcommand{\pp}{\,\mathrm{pp}}

\newcolumntype{Y}{>{\raggedright\arraybackslash}X}
\newcolumntype{P}[1]{>{\raggedright\arraybackslash}p{#1}}
\setlist{nosep,leftmargin=*}
\makeatletter
\patchcmd{\MaketitleBox}{\vskip36pt}{\vskip18pt}{}{\PackageError{submission}{Title spacing patch failed}{}}
\makeatother

\input{supplement-setup}

\begin{document}
\begin{frontmatter}
\title{A Reduced Stochastic Model of Microservice Availability: Algorithms for Identification and Formal Analysis}
\author[innopolis,mb3r]{Anatoly A. Krasnovsky\corref{corresponding}}
\ead{a.a.krasnovsky@gmail.com}
\author[innopolis]{Anna Maslovskaya}
\cortext[corresponding]{Corresponding author.}
\affiliation[innopolis]{organization={Innopolis University},city={Innopolis},country={Russia}}
\affiliation[mb3r]{organization={MB3R Lab},city={Innopolis},country={Russia}}

\begin{abstract}

Estimating microservice operation availability requires connecting runtime
observations to completion of the required work within a deadline.
Detailed architectural models provide one route, but their
reconstruction and analysis can be costly. We present algorithms for identifying
and analysing a reduced stochastic availability model from telemetry and
declared operation requirements. The identification procedure constructs a joint
law of execution states and retains compatible laws when observations are
incomplete. Formal results establish attainable probability bounds and
conditions for identifying a completion-and-deadline representation directly.
This construction preserves the range obtained by evaluating the same
completion-and-deadline event in the full model; a separate joint comparison
describes the structural information discarded. An independent study covered
three microservice applications and 18 campaigns. On complete common campaigns
with point forecasts from both methods, observed mean absolute errors differed
from those of independently extracted PMX/Palladio architectural models by at
most 0.082 percentage points. The full graph-based workflow from telemetry to first
forecast was 31.73--100.24 times faster than the studied PMX/Palladio workflow,
including startup and launcher waits. Direct identification preserved the
retained forecasts and intervals, with median paired first-answer savings of
5.1--11.2\% relative to full execution-model analysis. The results connect
required model detail to measured identification and analysis costs.
\end{abstract}
\begin{keyword}
stochastic modelling \sep availability \sep model identification \sep distributed tracing \sep microservices \sep model reduction \sep reliability analysis
\end{keyword}
\end{frontmatter}

\section{Introduction}\label{sec:introduction}
Estimating microservice operation availability requires a model of the work that makes an operation succeed. Traces reveal interactions, while request journals and replica observations describe their execution. Architecture-based approaches already turn such evidence into detailed models: the PMX extraction tool constructs Palladio Component Models (PCM), and the Palladio framework predicts reliability through architectural and Markov analysis~\cite{brosig2010,walter2017,weber2025,brosch2012}.
This paper develops algorithms for telemetry-driven identification and formal analysis of a reduced stochastic availability model. The work addresses model construction and reduction for modelling and simulation studies: it determines which execution details the requested availability probability requires and measures their identification and analysis costs. The reported probabilities and bounds are evaluated exactly.

Earlier resilience-model discovery studies recovered service graphs and simulated availability on DeathStarBench and the OpenTelemetry Demo~\cite{krasnovsky2026icse,krasnovsky2026aina}. They reproduced substantial observed behaviour but retained systematic discrepancies. The discovered graph supplies the starting structure. A reachable replica may never receive the selected call; a failed transport attempt may recover through a declared retry; a completed call may exceed the deadline. Formalizing these distinctions makes the graph's success abstraction testable and establishes which execution evidence can justify its refinement.

Three established lines of work provide the basis. Boolean structural reliability evaluates a success function under a joint law~\cite{colbourn1987,bryant1986}; model extraction connects observations to structure and parameters~\cite{brosig2010,weber2025}; multivalued observation maps and random-set identification characterize laws compatible with incomplete evidence~\cite{dempster1967,beresteanu2012}. The task here is to specify this connection for microservice operations and determine which distinctions the requested probability requires.

We investigate the hypothesis that a lightweight stochastic model can provide comparable availability estimates at lower cost than an independently extracted architectural model. 
Here \emph{lightweight} denotes a finite representation specialized to the requested probability; its detail is target-dependent, and its construction and analysis costs are measured.
Identification maps runtime evidence and operation declarations (success boundary, required call groups, retry policies and deadline) to a joint law over execution states, or to a family of compatible laws when the evidence is incomplete, using counts and exact computation.
Formal analysis establishes when a smaller representation can be identified directly from the same evidence, preserving the selected quantity while discarding other information. These results depend on evidence semantics, sound bindings and complete attempt accounting, including timeouts and missing traces.

The contribution has three connected parts:
\begin{enumerate}
\item \textbf{An operation model with explicit identification inputs.} The execution event, external-success correspondence and observation bindings connect a discovered graph and runtime records to an operation-level probability.
\item \textbf{Identification and analysis algorithms with formal reduction guarantees.} Joint attainability, point-identifiability and completion reduction justify the estimates. An image and lifting proof establishes equality of the reduced law families obtained by full-model projection and direct completion-and-deadline identification; a joint contrast between full and reduced events describes the information discarded.
\item \textbf{Validation of accuracy, necessary detail and cost.} Independent outcomes test semantics and forecasts; PMX/Palladio supplies an alternative identification-and-analysis route. Equivalent Boolean, Bayesian and Markov calculations isolate computational representation, while matched full/direct construction measures the cost of detail.
\end{enumerate}

Section~\ref{sec:related} reviews related work, Section~\ref{sec:model} presents the model, Sections~\ref{sec:identification}--\ref{sec:computation} develop identification and formal analysis, and Sections~\ref{sec:evaluation}--\ref{sec:results} report the experimental evaluation. The appendix provides experimental protocols, detailed results and analytical checks.

\section{Related work}\label{sec:related}
\subsection{Identification and extraction of executable models}
Model extraction from observations has substantial precedents. Process mining recovers behavioural structure~\cite{aalst2004}, simulation discovery adds executable parameters~\cite{rozinat2009}, and stochastic Petri-net discovery incorporates delay distributions and execution policies~\cite{rogge2014}. Reliability-oriented process mining derives manufacturing models from activity, resource and failure/repair records~\cite{friederich2022,friederich2025}. These approaches recover different model objects for different operational questions. Our target is a finite joint model of a declared microservice operation, including the observation-compatible states needed to identify its probability.

Architecture-level extraction connects monitoring data to performance and reliability analysis. Israr et al. recover layered queueing models from traces~\cite{israr2005}; Brosig et al. extract Palladio and related architectural models~\cite{brosig2010,brosig2011}. Extensible frameworks~\cite{walter2017}, Kieker~\cite{hasselbring2020}, incremental calibration~\cite{mazkatli2020} and CIPM~\cite{cipm2025} support construction and maintenance. PMX extracts PCM behaviour and parameters from traces and metrics, estimates failures from classified spans, and invokes Palladio prediction~\cite{weber2025}. It is the closest complete comparator: an independent evidence-to-model-to-estimate route. Our identification constructs an operation-specific joint law with explicit completion, delivery and timing bindings, then permits direct omission of detail proved unnecessary for the target. The comparison tests both routes' input interpretations, constructed objects and costs.

Input preparation contributes materially to model cost, as documented in SIMPAT's discussion of simulation data management~\cite{skoogh2012}. Dependent multivariate inputs~\cite{inputdependent2016} and multiformalism cloud modelling~\cite{cloudmulti2019} motivate preserving joint evidence while choosing its representation. Our cost boundaries include interpretation, construction and queries; common primitive bindings isolate the additional cost of full representation.

\subsection{Operation reliability and competing representations}
Dependability relates failures to a specified service~\cite{avizienis2004}; user-oriented and architecture-based reliability make execution, usage and component interaction explicit~\cite{cheung1980,reussner2003,gokhale2007,immonen2008}. Software performance engineering~\cite{woodside2007,koziolek2010}, Palladio~\cite{becker2009,brosch2012} and Descartes~\cite{huber2017} capture broader architectural and resource behaviour. The PMX/Palladio reliability path used here transforms PCM into a Markov model. Its application mappings and native implementation are fixed~\cite{pmx2026}; our comparison includes both extraction and analysis.

The nearest mathematical class is Boolean structural reliability. Network reliability evaluates connectivity under a state law~\cite{colbourn1987}; general counting complexity~\cite{valiant1979} motivates approximation and rare-event methods~\cite{karger1995,alban2017}. Fault trees express shared and dynamic failure structure~\cite{ruijters2015}; squarefree-polynomial analysis exploits bounded sharing in static fault-tree DAGs~\cite{sfpa2025}. BDDs share equal subfunctions, with size governed by variable order~\cite{bryant1986}. CUDD~\cite{cudd2026} therefore provides a direct test of our specialized Boolean evaluator.

Bayesian reliability models represent failure dependencies~\cite{boudali2005}; exact inference depends on elimination width and retained factors~\cite{darwiche2009,dechter1999,lazy1999}. aGrUM implements these inference facilities~\cite{agrum2020}. Markov reachability offers another exact representation, implemented by Storm~\cite{storm2022}. Proposition~\ref{prop:representations} translates the same identified task to auxiliary Bayesian support queries and Markov decision processes. These comparisons hold the observation law and target fixed, separating computational representation from independent model extraction.

Temporal reliability models identify additional dynamics. SIMPAT studies cover general repair and fault coverage~\cite{ke2010}, environment-dependent reliability~\cite{sinha2021}, correlated resource availability~\cite{javadi2018} and asynchronous transactions~\cite{mentis2012}; network-aware microservice placement adds deployment objectives~\cite{zhang2025}. Our state law describes attempts in a given regime. Failure/repair rates or predictions for changed deployments require corresponding model parameters and evidence.

\subsection{Identification, aggregation and uncertainty}
Network tomography studies identification from indirect observations~\cite{bu2002,duffield2006}. Missing-data inference depends on the observation mechanism~\cite{rubin1976}; preserving the attempt census and joint masks makes that mechanism part of our observed law. Dempster's multivalued observation maps induce lower and upper probabilities~\cite{dempster1967}. Random-set theory and partial identification characterize compatible distributions and expectations~\cite{beresteanu2012,manski2003}; sharp probability bounds also have classical optimization formulations~\cite{prekopa1990}. Applying those results requires a specified observation map and compatible-state restrictions. Here the operation contract and evidence bindings supply that map; Theorem~\ref{thm:joint} gives its finite completion reduction, and Proposition~\ref{prop:direct} establishes a direct construction of the reduced identified family.

Model aggregation has its own preservation requirements. Markov-chain lumpability concerns the dynamics retained by aggregated states~\cite{buchholz1994}. Knowledge compilation distinguishes representation size, construction cost and supported queries~\cite{compilation2002}. Our direct reduction preserves the attainable distribution of the completion-and-deadline image under an observation family. It can identify a target without identifying every state probability and can discard information required by another target. This is the distinction tested by the full/direct experiment.

Ergodic averaging supplies the consistency basis~\cite{gray2009}. Estimated simulation inputs introduce sampling uncertainty~\cite{barton2014}; confidence regions for partially identified quantities need an inferential construction~\cite{imbens2004}. Transfer methods impose source/target relationships~\cite{sugiyama2007}. Section~\ref{sec:uncertainty} distinguishes these effects from the range induced by incomplete observations.

\subsection{Runtime evidence and model validation}
Runtime models support monitoring and adaptation~\cite{blair2009,bencomo2019,garlan2004}, with uncertainty in the model--system relationship~\cite{giese2014} and explicit service boundaries~\cite{zimmermann2017}. X-Trace, Dapper and OpenTelemetry provide tracing mechanisms and conventions~\cite{fonseca2007,sigelman2010,oteltrace2026,otelhttp2026}; industrial experience documents practical instrumentation limits~\cite{li2022}. Our operation contracts and bindings specify which facts those records establish.

Verification checks computation, while validation tests the model in its intended experimental frame~\cite{sargent2013,foures2016}. Formal specifications guide systematic simulation tests~\cite{hollmann2014}, and interoperable simulation environments connect verification with practical modelling~\cite{sinisi2021}. We combine exact finite checks, equivalent solvers and independent operation outcomes. Fault injection~\cite{heorhiadi2016,basiri2016} and paired routing intervention~\cite{pearl2009} test mechanisms; explicit timing boundaries and repetition levels~\cite{georges2007,kalibera2013} test the claimed computational benefit.

\section{Formal operation and observation model}\label{sec:model}
\subsection{Operation contract and structural model}
For operation class $u$, let $Y_u$ indicate externally verified success and $A_u=\Prob(Y_u=1)$ in a declared regime. Its contract specifies attempt boundaries, response or persistence postconditions, mandatory logical call groups and deadline $d_u$. A logical group may contain transport attempts under a declared retry policy.

As in structural reliability, a Boolean success function is evaluated under a joint state law. The graph organizes its requirements. Identification specifies the state coordinates and the laws permitted by execution evidence, retaining compatible states for unavailable facts.

We formalize the operation as the stochastic model
\(\mathcal M_u=\langle G_u,\mathcal R,P,\Phi_u\rangle\). The finite directed graph $G_u=(V,\mathcal E,\tau)$ represents services and declared peers; $\tau(e)$ distinguishes synchronous and asynchronous interactions. $\mathcal R$ maps services to identified replicas. $\Phi_u$ contains entries, required targets, mandatory groups and observation bindings. $P$ is a joint state law, or a family of compatible laws when observations leave parameters unidentified.

Native cross-service parent relations and accepted database CLIENT peers supply edges; source/deployment declarations normalize identities. An unresolved required boundary produces a structural refusal. The contract states whether an asynchronous continuation belongs to the success boundary.

\subsection{Execution state and success event}
Let $X$ collect Boolean eligibility coordinates. Node, edge and replica guards can share coordinates, preserving common causes. For a controlled replica, eligibility means that a required invocation can reach that replica within the operation window. A previous health observation is contextual evidence, with its own timestamp and check semantics~\cite{haproxy2026}; it does not define this coordinate. The predicate $B_u(X)$ requires every declared target to be synchronously reachable from an appropriate entry through eligible nodes and edges.

Write the complete state as
\begin{equation}
 W=(Z,C),\qquad Z=(X,D,T),\qquad C=(C_0,\ldots,C_{r-1}).\label{eq:state}
\end{equation}
Here $D=(D_{s\rho})$ collects selection indicators: $D_{s\rho}=1$ records that a mandatory invocation selected replica $\rho$ of controlled service $s$. The selection footprint retains all required selected replicas in the attempt. There are $r$ completion coordinates: $C_0$ denotes entry completion and $C_j$, for $1\le j<r$, the completion of mandatory logical group $j$. The bit $T$ records satisfaction of the whole-operation deadline. Let $m=|Z|$ be the number of control coordinates; the complete state has $d=m+r$ coordinates.

Let $\mathcal V_u\subseteq V$ be the controlled service nodes required by the operation. For replica guard $a_{s\rho}(X)$, define
\begin{equation}
 J_s(X,D)=\left(\bigvee_{\rho\in\mathcal R(s)}D_{s\rho}\right)
 \land\bigwedge_{\rho\in\mathcal R(s)}\left(\neg D_{s\rho}\lor a_{s\rho}(X)\right).
\end{equation}
Every controlled service must have a mandatory selection, and every selected replica must satisfy its delivery guard. With $H_G$ indicating the presence of the required synchronous group edges, set
\begin{align}
 S_u(W)&=B_u(X)\land\bigwedge_{s\in\mathcal V_u}J_s(X,D),\label{eq:selection}\\
 C_u(W)&=H_G\land\bigwedge_{j=0}^{r-1}C_j,\label{eq:completion}\\
 E_u(W)&=S_u(W)\land C_u(W)\land T,\qquad
 A_u^{\mathcal M}=\E_P E_u(W).\label{eq:execution}
\end{align}
The outer conjunction concerns mandatory logical groups. A group's binding resolves its declared internal retry policy; uncertainty about its final outcome remains uncertainty in $C_j$. Optional work is excluded from the conjunction, while its elapsed time remains inside the deadline. The implemented adapters observe two replicas of one controlled service and fix the eligibility of other nodes true.

\subsection{Correspondence to external success}\label{sec:adequacy}
Let $h(\omega)$ map a concrete execution to a complete model state. The correspondence $Y_u(\omega)=E_u(h(\omega))$ requires the following operational conditions:
\begin{enumerate}[label=(A\arabic*)]
\item The contract includes all required work and prerequisites. Successful execution reaches each required target through an admitted synchronous causal path. Recovery within a mandatory logical group is covered by its declared policy.
\item Completion bits faithfully express the group's local postconditions and source-level propagation. Their conjunction is necessary and sufficient for the required external result before applying the deadline. The timing boundaries agree.
\item Successful mandatory execution reaches its selected replicas within the operation window and satisfies the interpreted eligibility guards.
\end{enumerate}
Under A1--A2, external success is equivalent to $T\prod_jC_j=1$. A1 and A3 additionally establish $H_G$, $B_u$ and every $J_s$ on successful executions. Conversely, $E_u=1$ entails the required completions and timing. Thus
\begin{equation}
 Y_u=E_u\circ h.\label{eq:adequacy}
\end{equation}
The application contract and source bindings must establish these conditions. In particular, a required database effect needs evidence beyond a response acknowledgement.

The structural abstraction satisfies
\begin{equation}
 E_u\le S_u\le B_u,\qquad
 \E B_u-\E E_u=\Prob(B_u=1,E_u=0).\label{eq:gap}
\end{equation}
The difference is precisely the mass of reachable states that fail an execution requirement. This decomposition motivates the same-law ablations and routing intervention.

\subsection{Partial observations and sound bindings}\label{sec:binding}
For each attempt, let $M\subseteq\{1,\ldots,d\}$ index the observed state coordinates and let $\widetilde W_M$ contain their recorded values. Every attempt, including timeouts and missing traces, supplies $O=(M,\widetilde W_M)$. A sound binding records $\widetilde W_M=W_M$; missing or conflicting evidence remains unknown. Its compatible-state set is the Cartesian fiber
\begin{equation}
 \mathcal F(o)=\{w\in\bits^d:w_M=o_M\}.\label{eq:fiber}
\end{equation}
The mask can depend on execution and failure. Keeping values and masks jointly avoids a missing-at-random requirement~\cite{rubin1976}. An incorrectly recorded known bit can exclude the true state from the fiber; masking an unavailable bit does not.

An identified native SERVER or declared source entry between request start and the earlier of request end and deadline establishes delivery, even if the application returns an error. Otherwise eligibility stays unknown. Traces bind selection and logical completion; the attempt journal binds $T$. Historical proxy verdicts remain contextual; marginalizing their joint counts introduces no independence assumption.

The timestamp comparison presumes a common time base. Across hosts, a justified bound on relative timestamp error must place the entire possible entry-time interval inside the request/deadline window; otherwise delivery remains unknown. Causal span ordering alone does not certify that window~\cite{sigelman2010}. The single-runner confirmation tests this binding on a shared host time base, without validating cross-host synchronization.

\subsection{Running example: a visit-creation operation}\label{sec:example}
A gateway calls a two-replica Visits service with one declared POST retry after a 503 response. Replica A was previously reported down. The first transport attempt receives 503; the second receives 201 with an identified SERVER entry at A. That entry establishes delivery despite the earlier health verdict. Sequential CLIENT spans with the same parent, explicit final success and one receiver for the successful retry establish logical completion. The first transport attempt remains failed.

Figure~\ref{fig:running-example} connects the operation structure, execution evidence and compatible states. In the illustrated case, only replica A is selected and all other required conditions are fixed at one, so $E_u=T C_{\mathrm{Visits}}$. With $T=1$ and known completion, unknown eligibility of unused replica B cannot change the outcome. Masking the completion coordinate admits both success and failure.

\begin{figure}[!htbp]
\centering\includegraphics[width=\linewidth]{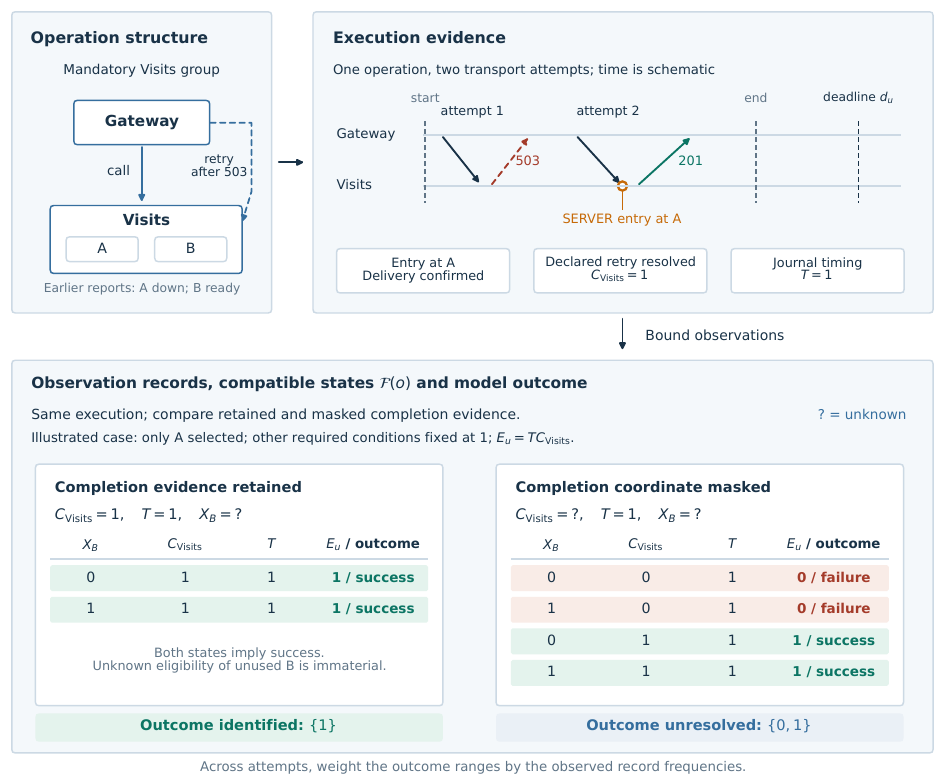}
\caption{Visit-creation example. The dashed structural arrow denotes the declared retry after a 503 response. A SERVER entry establishes delivery to A, while the attempt journal establishes deadline satisfaction. Here $C_{\mathrm{Visits}}$ is the completion bit of the mandatory Visits group. The lower tables compare the same execution with completion evidence retained or masked; all unshown required conditions are fixed at one. Across attempts, compatible outcome ranges are weighted by observed record frequencies.}\label{fig:running-example}
\end{figure}

Unresolved receiver, order or possible first-attempt effects leave completion unknown. If all other requirements are established, this gives outcome bounds $[0,1]$; a known deadline failure instead forces zero. An independent SQL read verifies the required write effect for $Y_u$, without importing that validation label into the observation binding.

\FloatBarrier
\section{Identification from runtime observations}\label{sec:identification}
\subsection{Inputs, construction and target choice}
The inputs are runtime traces, a complete attempt journal, available replica observations, and the operation/source/deployment declarations. The declarations specify stable identities, required logical groups, permitted recovery and success boundaries. Independent response and persistence labels are reserved for validation. Identification concerns the joint execution law in the observed workload and failure regime.

The procedure has five stages:
\begin{enumerate}
\item \emph{Establish the population and interpret evidence.} Associate native records with each attempt, normalize identities and apply the source bindings. Retain timeouts and attempts without usable traces in the denominator; unavailable coordinates remain unknown.
\item \emph{Determine required detail.} For the full execution target and structural contrasts, recover the observable service graph, replica selections and delivery guards. For the completion-and-deadline target alone, the direct construction in Section~\ref{sec:direct} uses the declared logical groups and their evidence without constructing the full state.
\item \emph{Check support and bind observations.} Check the boundaries needed by the selected representation. Unexplained required native-parent relations cause refusal; absent evidence for an otherwise supported coordinate gives a partial state. Save the interpretation associated with each coordinate.
\item \emph{Identify the observation law.} Aggregate equal partial observations into integer counts $n_o$ and set $\widehat Q(o)=n_o/n$. Values and missingness masks form one category, so their dependencies are retained.
\item \emph{Evaluate and retain the result.} Compute attainable target extrema under that law. Save the representation, category counts, bounds or point, support status and input identity for replay. Forecasts use calibration observations only.
\end{enumerate}

Stages one and three define $\mathcal F(o)$, stage two selects the representation, and stage four supplies $\widehat Q$. In the visit example, established retry completion and timing enter the observation category; the SQL validation label remains separate. Reduction follows these primitive bindings.

\begin{figure}[htbp]
\centering\includegraphics[width=\linewidth]{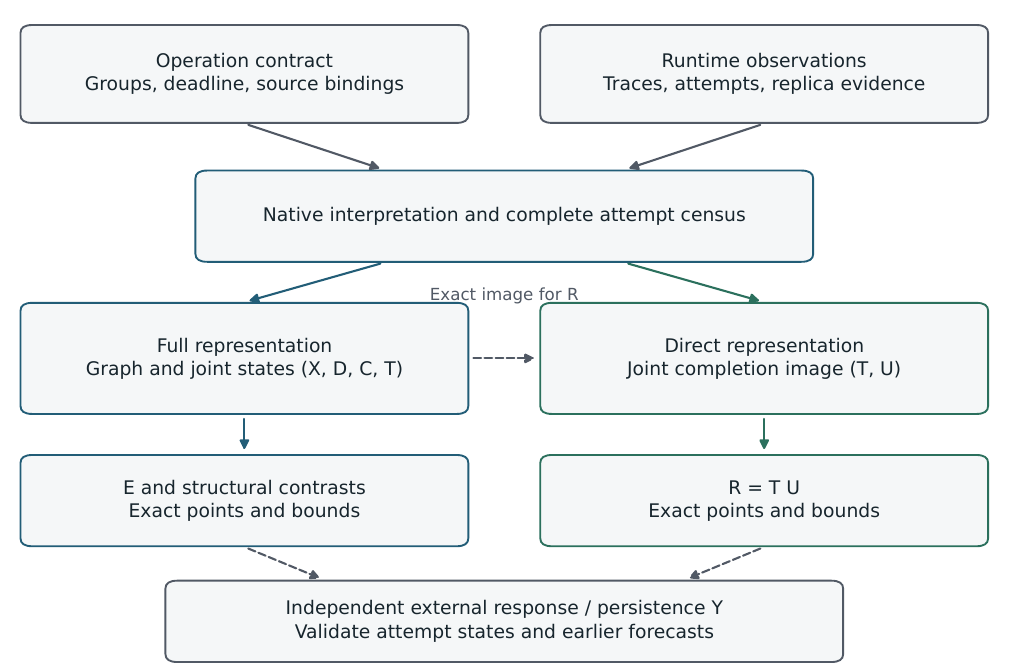}
\caption{Schematic of identification under a declared operation contract. Full and direct routes share interpretation rules. Formal projection relates their completion-and-deadline targets; separate external outcomes validate the model and earlier forecasts.}\label{fig:method}
\end{figure}

\subsection{The identified family}
For population observation law $Q$, the compatible joint laws have the form
\begin{equation}
 \Prob(W=w,O=o)=Q(o)\pi_o(w),\qquad
 \pi_o\in\Delta(\mathcal F(o)),\label{eq:family}
\end{equation}
where $\Delta(S)$ is the probability simplex on finite set $S$. Their complete-state marginals are $P(w)=\sum_oQ(o)\pi_o(w)$. Replacing $Q$ by $\widehat Q$ yields the empirical identified model. With complete observations each fiber is a singleton and the construction gives the empirical joint state law. With incomplete observations it retains all conditional distributions allowed within each fiber. Its estimated parameters remain joint category frequencies; individual physical probabilities need not be identified for a target probability to be identified.

The full implementation evaluates the six Boolean predicates
$B_u$, $S_u$, $S_uC_u$, $B_uC_uT$, $S_uT$ and $E_u$,
together with $f-E_u$ for each of the first five predicates $f$.
Let $F=(f_1,\ldots,f_p)$ collect these eleven quantities.
Every quantity is evaluated on the same complete state; each difference is Boolean because $E_u$ implies the other predicate. This common law supports both the success estimate and analysis of the requirements that separate reachability from execution.

\FloatBarrier
\section{Analytical properties of the identified model}\label{sec:formal-analysis}
\subsection{Joint attainable values and completion reduction}\label{sec:joint}\label{sec:reduction}
Put $c(C)=\prod_{j=0}^{r-1}C_j$. Every requested functional, and every linear contrast of the vector, factors through
\begin{equation}
 F(Z,C)=a(Z)+b(Z)c(C),\label{eq:affine}
\end{equation}
with vector-valued coefficients; $H_G$ is absorbed into them. For category $o$, let $\alpha_o$ be the product of known completion bits and $r_o$ the number of unknown completion bits. The possible extreme conjunction values are
\begin{equation}
 c_o^+=\alpha_o,\qquad
 c_o^-=\begin{cases}\alpha_o,&r_o=0,\\0,&r_o>0.\end{cases}\label{eq:completion-extrema}
\end{equation}
Both are attainable by retaining known bits and setting every unknown completion to one or zero. Let $\mathcal R_o$ contain those two completed states for every compatible control assignment $z$; duplicates may be removed.

The observation map $o\mapsto\mathcal F(o)$ is a finite random set. Lower and upper probabilities induced by such maps originate with Dempster~\cite{dempster1967}; attainable expectations and identification regions have established random-set and optimization formulations~\cite{beresteanu2012,manski2003,prekopa1990}. The following theorem gives their finite specialization and the model-specific completion-image reduction. Convexification here follows directly from the unrestricted conditional mixtures in~\eqref{eq:family}.

\begin{theorem}[Exact joint identification by completion reduction]\label{thm:joint}
Fix a finite model, Cartesian observation fibers~\eqref{eq:fiber}, the family~\eqref{eq:family} and a target vector satisfying~\eqref{eq:affine}. The set of compatible expectation vectors is exactly
\begin{equation}
 \mathcal I_Q(F)=\sum_{o:Q(o)>0}Q(o)
       \operatorname{conv}\{F(w):w\in\mathcal R_o\}.\label{eq:jointset}
\end{equation}
For every $\lambda\in\mathbb R^p$, the lower and upper values of $\lambda^\top\theta$ over $\theta\in\mathcal I_Q(F)$ are
\begin{align}
 L_Q(\lambda^\top F)&=\sum_o Q(o)\min_{w\in\mathcal R_o}\lambda^\top F(w),\label{eq:lower}\\
 U_Q(\lambda^\top F)&=\sum_o Q(o)\max_{w\in\mathcal R_o}\lambda^\top F(w).\label{eq:upper}
\end{align}
Each scalar endpoint is attained by a compatible joint law. A scalar target is identified if and only if it is constant on every positive-$Q$ fiber. If $k_o$ controls are unknown in category $o$, at most $2\sum_o2^{k_o}$ complete-state evaluations suffice for any fixed set of scalar targets and contrasts.
\end{theorem}

\begin{proof}
Conditional on $O=o$, every expectation of $F(W)$ lies in the convex hull of $F(\mathcal F(o))$. Conversely, a convex combination of those finitely many vectors supplies a conditional law $\pi_o$. Choosing these laws separately across categories preserves $Q$, so their weighted Minkowski sum is exactly the compatible expectation set.

For a fixed compatible $z$, the vector in~\eqref{eq:affine} depends on $C$ only through the Boolean value $c(C)$. Every feasible such value is either $c_o^-$ or $c_o^+$ and has a compatible uniform completion witness. Consequently $F(\mathcal F(o))=F(\mathcal R_o)$, preserving the entire convex hull, including correlations among requested targets.

A linear functional reaches its extrema at the generating vectors of each finite convex hull. The conditional law concentrated at a minimizing or maximizing witness in each category attains~\eqref{eq:lower} or~\eqref{eq:upper}; mixtures attain intermediate scalar values. The width is a sum of nonnegative within-fiber ranges weighted by $Q(o)$, and vanishes exactly when all positively weighted ranges vanish. There are $2^{k_o}$ compatible control assignments and at most two retained completion states for each.
\end{proof}

Endpoint attainability is relative to the declared family~\eqref{eq:family}. It does not assert that every extremal state is physically realizable by a particular deployment. Adding physical cross-coordinate restrictions changes that family and requires new feasibility checks.

The theorem preserves joint feasibility, which separate marginal intervals can lose. For example, optimizing $B_u-E_u$ directly respects their occurrence on one state; subtracting independently optimized bounds need not do so. The implementation computes scalar extrema, contrasts and witnesses. Equation~\eqref{eq:jointset} characterizes the associated joint set without requiring its explicit polytope construction.

The structural hypotheses matter. For the completion expression $C_1(1-C_2)$, both uniform completions evaluate to zero while $(1,0)$ evaluates to one. Additional cross-coordinate constraints can also invalidate a uniform completion witness. Such models require a reduction that preserves their own feasible target image.

\subsection{Decisive observations and dependence}\label{sec:certificates}
For execution, write $E_u=g_u(Z)c(C)$ with $g_u=H_GS_uT$, and let $g_o^-,g_o^+$ be its control extrema. The Cartesian fibers give
\begin{equation}
 \ell_o=g_o^-c_o^-,\qquad u_o=g_o^+c_o^+,\qquad
 U_Q(E_u)-L_Q(E_u)=\sum_o Q(o)\ind\{\ell_o<u_o\}.\label{eq:certificates}
\end{equation}
A known failed mandatory completion certifies failure even with unknown selection. Success requires every compatible control assignment to succeed and all mandatory completions to be known true. Every remaining ambiguous category contributes its full mass to the interval width. A probability can thus be identified before every coordinate is identified.

Joint selection is essential to this calculation. Consider two mandatory calls, an unavailable replica A and an available replica B. Equally likely route words AA/BB yield success probability $1/2$; equally likely AB/BA yield zero. Both laws give identical per-call routing marginals. Their joint footprints distinguish operation success.

The algorithm aggregates integer-weighted extrema under $\widehat Q$, returns a point when they coincide, and saves the graph, bindings, counts and results for independent replay. This is exact analysis of the empirically identified family; Sections~\ref{sec:consistency} and~\ref{sec:uncertainty} connect it to population and external-operation quantities.

\subsection{Direct identification of completion-and-deadline availability}\label{sec:direct}
The completion reduction in Theorem~\ref{thm:joint} retains all controls and the full target vector. A further specialization is possible when the requested event is
\begin{equation}
 U=\prod_{j=0}^{r-1}C_j,\qquad R_u=T U,\qquad
 K_u=H_GS_u,\qquad E_u=K_uR_u.\label{eq:direct-target}
\end{equation}
Under A1--A2, $R_u$ expresses the declared external result; A3 supplies the additional full-event correspondence. Direct identification of $R_u$ and replacement of $E_u$ by $R_u$ are distinct operations.

For a partial observation $o$, define $\bar U(o)=0$ if any known completion is false, $\bar U(o)=1$ if every completion is known true, and $\bar U(o)=?$ otherwise. Write $\bar o=(o_T,\bar U(o))$ and aggregate its probability as $\bar Q(b)=\sum_{o:\bar o=b}Q(o)$. The reduced observation fiber $\mathcal F_R(b)$ fixes its known coordinates in $\bits^2$.

\begin{proposition}[Exact target image and direct identification]\label{prop:direct}
Fix the attempt population, logical groups and primitive observation bindings. Under the Cartesian family~\eqref{eq:family}, the map $q(W)=(T,U)$ satisfies
\begin{equation}
 q(\mathcal F(o))=\mathcal F_R(\bar o).\label{eq:direct-image}
\end{equation}
The family of distributions of $(T,U)$ obtained by projecting the full model is exactly the family identified from $(\bar Q,\mathcal F_R)$. Consequently, full-model evaluation of $R_u$, explicit projection followed by evaluation, and direct construction from the same primitive observations yield the same attainable $R_u$ probabilities.

Let $n_1$ count attempts whose reduced observation forces $TU=1$, and $n_?$ those permitting both values. The empirical bounds are
\begin{equation}
 [\widehat L(R_u),\widehat U(R_u)]
 =\left[\frac{n_1}{n},\frac{n_1+n_?}{n}\right].\label{eq:direct-bounds}
\end{equation}
The reduced representation has two coordinates and at most nine partial-observation categories, independently of the number of mandatory groups.
\end{proposition}

\begin{proof}
If a known completion is false, its conjunction is fixed false. If all are known true, the conjunction is fixed true. Otherwise no known completion is false and at least one is unknown: setting all unknown completions true realizes $U=1$, while setting one false realizes $U=0$. The deadline is a separate coordinate and can take each value admitted by $o_T$ in the Cartesian fiber. Other coordinates can be filled consistently with their known values, proving~\eqref{eq:direct-image}.

Projection of any compatible full law therefore supplies a compatible reduced law. Conversely, for each category and each reduced state, choose a full-state preimage as above and assign it the desired conditional mass. Categories with the same reduced mask have identical reduced fibers. Assigning the same reduced conditional distribution to all of them lifts any law after their weights are merged. Thus projection and direct identification have the same range of probability laws on $(T,U)$. Applying the target $TU$ gives~\eqref{eq:direct-bounds}; its unknown outcomes can independently be assigned to either endpoint within each observation category.
\end{proof}

Lifting establishes that direct identification admits exactly the projected laws after category merging. The full semantics therefore justifies constructing the smaller representation first. The structural contrast below identifies information outside that retained family.

The direct route parses original records and binds the required completions and deadline before aggregation. It shares the full route's native-parent and logical-retry interpretation, but constructs neither the complete replica state nor a full model to project. It has its own support check: a missing declared completion may remain unknown without requiring a complete graph, whereas an unexplained native call boundary still prevents a justified binding. Comparisons between routes use their common supported population and report other support separately. A probability of $R_u$ can be identified even when the joint law of $(T,U)$ is not, for example when every attempt has $T=0$ and $U$ remains unknown.

\subsection{When the additional structural conditions matter}\label{sec:detail}
On a specified domain $\mathcal H$, $E_u=R_u$ exactly when $\mathcal H$ contains no state with $R_u=1,K_u=0$. This can hold on actual executions while the larger observation fibers still admit such states. The joint diagnostic
\begin{equation}
 D_u=R_u-E_u=R_u(1-K_u),\qquad
 [L_Q(D_u),U_Q(D_u)]\label{eq:direct-gap}
\end{equation}
is evaluated on the original common family. Since $D_u\ge0$, $U_Q(D_u)=0$ establishes equality for every compatible law; $L_Q(D_u)>0$ forces a positive difference; $L_Q(D_u)=0<U_Q(D_u)$ leaves the difference unresolved.

Table~\ref{tab:detail-example} follows two mandatory completions through the full and reduced events. In row B, incomplete completion evidence leaves success unresolved but the events agree throughout the fiber. In row D, their separate intervals again coincide, yet a structural difference is possible. Row C shows a determinate reduced event alongside an unresolved full event. These distinctions motivate the joint comparison in the experiment.

\begin{table}[htbp]
\centering\small
\caption{Illustrative Cartesian fibers with two mandatory completions and a structural guard $K$. Each question mark is free within the fiber; $U=C_1C_2$, $R=TU$, $E=KR$ and $D=R-E$. Singletons are printed as points. The same primitive completion/deadline observations identify $R$ directly or through the full model.}\label{tab:detail-example}
\setlength{\tabcolsep}{5pt}
\begin{tabular}{@{}cccccclll@{}}
\toprule
Case & $T$ & $C_1$ & $C_2$ & $K$ & $\bar U$ & $I(R)$ & $I(E)$ & $I(D)$\\\midrule
A & 1 & 1 & 1 & 1 & 1 & $1$ & $1$ & $0$\\
B & 1 & 1 & ? & 1 & ? & $[0,1]$ & $[0,1]$ & $0$\\
C & 1 & 1 & 1 & ? & 1 & $1$ & $[0,1]$ & $[0,1]$\\
D & 1 & 1 & ? & ? & ? & $[0,1]$ & $[0,1]$ & $[0,1]$\\
E & 0 & ? & ? & ? & ? & $0$ & $0$ & $0$\\\bottomrule
\end{tabular}
\end{table}

The reduction preserves the selected completion-and-deadline probability. It discards structural contrasts: two laws can have $T=U=1$ throughout but set $K_u$ always true or always false, giving identical reduced laws and different $E_u$ and $D_u$. The full representation retains these distinctions. If physical cross-coordinate constraints are added, direct identification must use their exact projected fibers; the two-coordinate Cartesian image is justified by the observation family specified here.

\subsection{Consistency of construction and estimation}\label{sec:consistency}
Structural discovery and numerical identification form one statistical procedure. Their consistency requires observability of the graph and stable interpretation of its records.

\begin{proposition}[Discovery-and-estimation consistency]\label{prop:consistency}
Let $G_*$ be a finite observable operation graph on stable declared identities, satisfying the required contract boundaries. Suppose that accepted edge evidence is sound, every edge of $G_*$ has positive probability of observation, and the joint sequence of evidence and attempts is stationary and ergodic. Assume a fixed finite observation alphabet and a binding rule that stabilizes when the graph stabilizes; earlier bindings are recomputed in batch, or only finitely many earlier online bindings differ. Then the cumulative discovered graph equals $G_*$ eventually almost surely. For every fixed target in Theorem~\ref{thm:joint},
\begin{equation}
 (\widehat L_n(f),\widehat U_n(f))
 \longrightarrow (L_Q(f),U_Q(f))\quad\text{almost surely}.\label{eq:consistency}
\end{equation}
\end{proposition}

\begin{proof}
For each edge $e$, the ergodic average of its accepted-observation indicator converges to $p_e>0$~\cite{gray2009}. Its first observation time is therefore finite almost surely. The maximum of these times over finite $E_*$ is finite, and soundness excludes extraneous edges. The cumulative graph consequently stabilizes at $G_*$.

After stabilization, the fiber extrema $\ell_f(o),u_f(o)$ are fixed bounded functions of the finite observation alphabet. Their empirical averages are exactly the identified endpoints, so ergodic averaging proves~\eqref{eq:consistency}. Rebinding retained observations gives the same batch averages; a finite earlier online prefix has vanishing weight.
\end{proof}

The proposition concerns the graph observable in the declared regime. Permanently uninstrumented edges violate its coverage hypothesis. The limit is the identified range, which may remain non-singleton. A finite-sample singleton can also precede the first observation of a rare ambiguous category.

\subsection{Information refinement and the frequency limit}\label{sec:frequency}
A sound observation refinement splits each old fiber into subsets. On the same attempt population and with the same target, conditional averaging of the refined extrema gives
\begin{equation}
 L_{\rm old}(f)\le L_{\rm refined}(f)
 \le U_{\rm refined}(f)\le U_{\rm old}(f).\label{eq:refinement}
\end{equation}
Equation~\eqref{eq:certificates} identifies where extra observations can help: only ambiguous categories can reduce their contribution to width.

With complete states, identification reduces to
\begin{equation}
 \widehat A_u^{\mathcal M}=\frac1n\sum_{i=1}^n E_u(W_i).\label{eq:frequency-limit}
\end{equation}
Under~\eqref{eq:adequacy}, this is the external-success frequency $\widehat B_0$ on those attempts. The formal model additionally exposes the event structure, evaluates related functionals under one law and supplies bounds when some execution coordinates are unavailable.

\subsection{Sampling, binding and predictive adequacy}\label{sec:uncertainty}
Let $Q_s$ be the source observation law and $P_s,P_t$ source and target complete-state laws on a fixed model. Define
\begin{align}
 \varepsilon_n&=\|\widehat Q_n-Q_s\|_{\rm TV}, &
 \delta_s&=\Prob_s\{W\notin\mathcal F(O)\},\nonumber\\
 \rho&=\|P_t-P_s\|_{\rm TV}, &
 \eta_t&=\Prob_t\{Y_u\ne E_u(W)\}.\label{eq:error-terms}
\end{align}
They measure sampling error, unsound known-bit bindings, regime change and semantic discrepancy. With $b_n=\varepsilon_n+\delta_s+\rho+\eta_t$,
\begin{equation}
 A_{u,t}\in[\max\{0,\widehat L_n-b_n\},
             \min\{1,\widehat U_n+b_n\}].\label{eq:robust-containment}
\end{equation}
Indeed, fiber extrema bracket $E_u$ on sound bindings; the remaining mass changes either inequality by at most $\delta_s$. Bounded-function total-variation inequalities account for $\varepsilon_n$ and $\rho$, and indicator disagreement accounts for $\eta_t$. The supplement gives the derivation and serial-dependence calculation.

The returned identification interval is $[\widehat L_n,\widehat U_n]$. Equation~\eqref{eq:robust-containment} describes the additional quantities needed for a predictive envelope. Sampling confidence regions require a further inferential construction~\cite{imbens2004,barton2014,inputdependent2016}. With no restriction connecting source and target laws, target availability has only the range $[0,1]$. These distinctions determine the roles of independent external validation and campaign-level resampling.

\section{Computational structure and competing formalisms}\label{sec:computation}
\subsection{On-demand solution and reusable preparation}\label{sec:complexity}
For already interpreted observations, let $n$ be the number of aligned attempts, $K$ the number of distinct masked categories, $d=m+r$ the state dimension and $k=\max_o k_o$ the largest number of unknown controls. Let $b_G$ bound the work of forming a state and evaluating the fixed target vector on the graph. Theorem~\ref{thm:joint} gives the on-demand bound
\begin{equation}
 T_{\rm demand}=O\!\left(nd+Kd\log K+
                 b_G\sum_o2^{k_o}\right)
          =O(nd+Kd\log K+K2^k b_G).\label{eq:complexity}
\end{equation}
The first terms cover category counting and canonical sorting. Native-record decoding, identity resolution and source interpretation precede this bound and enter the complete identification measurements. Unknown mandatory completions no longer enter the exponent. The implementation accepts at most ten control and 64 total coordinates. Batch storage is $O(nd+Kd+2^kd+|G|)$ for the fixed eleven outputs, retaining attempts, category certificates and the current candidate states.

Reusable preparation compiles the targets into a shared Boolean circuit and tabulates $L=2^{m+1}$ assignments to the controls and one completion bit. Variable and target bit sets select each category's known-value region. Emptiness against the target and its complement gives the extrema; category counts supply their weights. New counts and masks reuse the representation.

Every solver receives the same completion projection: an observed false completion fixes $U=0$, all observed true fixes $U=1$, and otherwise $U$ is unknown. Joint control observations and weights are retained. Merging identical projected categories preserves the target set by Theorem~\ref{thm:joint}. Both original and projected inputs are measured, including conversion time.

Let $c_G$ include circuit construction and one evaluation, $p$ denote the fixed number of targets, $h_L=\lceil L/\omega\rceil$ machine words per bit set, and $b_{\rm key}$ cover semantic-key construction and validation. The implemented immutable-integer bit sets have the conservative bounds
\begin{align}
 T_{\rm prepare}&=O\bigl(Lc_G+(m+p)Lh_L\bigr),\label{eq:prepared-build}\\
 T_{\rm query}&=O\bigl(b_{\rm key}+Kd+K(m+p)h_L\bigr).\label{eq:prepared-query}
\end{align}
The second preparation term covers allocation and copying of growing integers; storage is $O(c_G+Kd+(m+p)h_L)$ words. Preparation depends on all $m$ controls, versus the $k_o$ unknown controls in on-demand enumeration. A table can therefore accelerate reuse yet cost more to construct when few controls are missing. Circuit size also depends on the reachability encoding.

\subsection{Cost of the direct target representation}
For $R_u$, Proposition~\ref{prop:direct} replaces the $d$ state coordinates by $(T,U)$. After primitive completion and deadline interpretation, direct category construction requires $O(nr)$ work and at most nine counters; a query to the identified reduced law takes constant arithmetic work. Full-model evaluation of the same $R_u$ already avoids unrelated controls and costs $O(Kr)$ on its retained categories. Explicit projection costs $O(Kr)$ once and leaves the same constant-size query representation as direct identification.

Let $B_R(n,S)$ denote reading and source interpretation for $S$ native records under the operation declarations. Direct construction costs $B_R(n,S)+O(nr)$, plus saving. Parsed records and per-attempt evidence remain in memory, so a small probability representation does not imply constant process memory. Full construction additionally forms structural and replica observations. Stage measurements separate these costs, the final query and interpreter startup.

\subsection{Exact symbolic, Bayesian and Markov representations}\label{sec:representations}
Ordered binary decision diagrams (BDDs) provide the closest mathematical alternative: both approaches evaluate Boolean structure under a specified observation law. BDDs share equivalent subfunctions instead of enumerating a control table~\cite{bryant1986,compilation2002}. For $E=g(Z)\prod_jC_j$, an ordinary ordered BDD placing controls first can attach one shared completion chain to the true terminal of $g$. Its nonterminal count is at most $B_g+r$, where $B_g$ is the BDD size of $g$ under the chosen control order. This existing completion sharing motivates a matched comparison. CUDD~\cite{cudd2026} retains its manager across structural changes, with fixed ordering or sifting at semantic rebuilds.

Bayesian and Markov comparisons must preserve the freedom in~\eqref{eq:family}:

\begin{proposition}[Equivalent support and reachability queries]\label{prop:representations}
For a finite masked law $Q$ with nonempty Cartesian fibers and a Boolean target $f$:
\begin{enumerate}[label=(\roman*)]
\item Construct an auxiliary Bayesian network whose root $O$ has any strictly positive distribution on the support of $Q$, with deterministic known coordinates given $O$, conditionally independent fair unknown coordinates, and a deterministic circuit for $f$. Then
\begin{align}
 L_Q(f)&=\sum_oQ(o)\ind\{\Prob_{\rm aux}(f=0\mid o)=0\},\label{eq:bn-lower}\\
 U_Q(f)&=\sum_oQ(o)\ind\{\Prob_{\rm aux}(f=1\mid o)>0\}.\label{eq:bn-upper}
\end{align}
\item Construct a finite decision process that first chooses $o$ with probability $Q(o)$, fixes known coordinates, and lets a scheduler choose unknown coordinates before reaching a terminal state labelled by $f$. Its minimum and maximum reachability probabilities equal $L_Q(f)$ and $U_Q(f)$. If the target is fixed on every fiber, the decision process can be collapsed to a labelled discrete-time Markov chain (DTMC) with that success probability.
\end{enumerate}
For the vector in Theorem~\ref{thm:joint}, replacing completion assignments by their attainable conjunction values preserves the joint terminal-label image in these constructions.
\end{proposition}

\begin{proof}
Conditional on a positive-mass category, the auxiliary network assigns strictly positive probability to every compatible complete state and zero to every incompatible state. Hence its failure posterior is zero exactly when the target is true throughout the fiber, and its success posterior is positive exactly when some compatible state succeeds. These are the category extrema in~\eqref{eq:lower}--\eqref{eq:upper}.

In the decision process, every scheduler produces a conditional distribution supported on the compatible states. Conversely, any such distribution can be realized through successive conditional choices along the coordinate sequence, using a randomized history-dependent scheduler. The initial probabilistic choice fixes $Q$. Optimizing a terminal Boolean reward therefore optimizes over precisely the family~\eqref{eq:family}; deterministic choices of extremal states suffice for each scalar endpoint. The completion reduction preserves the feasible target vectors by~\eqref{eq:affine}, so every reduced terminal vector has a complete-state witness and every original vector has a reduced representative.
\end{proof}

The auxiliary Bayesian law is a support oracle; the empirical weights $Q(o)$ enter the final exact aggregation. aGrUM's ordinary-network LazyPropagation~\cite{agrum2020,lazy1999} supplies either category-conditioned marginals or joint category--target tables, whose zero and positive cells give the same extrema. Storm~\cite{storm2022} supplies exact rational DTMC/MDP reachability on the decision graph or a compact representation that merges equal continuation structures and bypasses deterministic steps. 

\subsection{Comparison by task and computational parameter}
Table~\ref{tab:formalism-cost} separates PCM's independent identification and Markov analysis from CUDD/aGrUM/Storm translations of our fixed identified task.

\begin{table}[htbp]
\centering\footnotesize
\caption{Reliability representations and scoped computational costs. The Boolean/BDD row is the nearest mathematical comparison; the PCM/Palladio row is the nearest complete extraction-and-analysis comparison.}\label{tab:formalism-cost}
\begin{tabularx}{\linewidth}{@{}P{0.22\linewidth}P{0.29\linewidth}Y@{}}
\toprule
Class and representative & Task and assumptions & Cost driver\\
\midrule
Static fault trees; bottom-up analysis, SFPA~\cite{ruijters2015,sfpa2025} & Independent basic-event law; read-once tree or shared-event DAG & $O(N)$ for a read-once tree; SFPA admits $O(N(q+1)4^q)$ for $q$ multiparent nodes.\\
\textbf{Boolean structure; CUDD BDDs}~\cite{bryant1986,cudd2026} & Same masked law; satisfiability of each target and complement within a category cube & On an existing BDD, restriction and traversal take $O(KpB)$; cube-based Apply has a conservative $O(KpdB)$ bound. Compilation can be exponential and order-sensitive.\\
Bayesian networks; aGrUM~\cite{agrum2020,dechter1999} & Exact support queries via Proposition~\ref{prop:representations}; conditioned marginals or joint targets & Conditioned dense elimination: $O(KN a^{w+1})$ for $N$ variables, cardinality $a$ and induced width $w$. Joint targets trade repeated inference for retained factors.\\
DTMC/MDP; Storm~\cite{storm2022} & Terminal-label reachability in the equivalent decision graph & Acyclic backward evaluation admits $O(p(s+e))$ after constructing $s$ states and $e$ transitions. General models use solver-dependent linear systems or optimization.\\
\textbf{Architectural reliability (Markov); PMX, PCM, Palladio}~\cite{brosch2012,weber2025} & Independently extracted control flow and failures; Markov success reachability & Extraction plus transition-system construction and solution. Dense elimination admits $O(s^3)$ arithmetic work; actual cost depends on reduction and sparsity.\\
Present on-demand solver & Joint masked law; completion conjunction & $O(nd+Kd\log K+K2^kb_G)$; no reusable compilation.\\
Present prepared solver & Same law and targets; reusable control table & Preparation~\eqref{eq:prepared-build}; query~\eqref{eq:prepared-query}. Exponent depends on all $m$ controls.\\
Direct completion target & Primitive completion/deadline observations; exact image in Proposition~\ref{prop:direct} & $O(nr)$ aggregation after binding; at most nine categories and constant arithmetic work per ready query.\\
\bottomrule
\end{tabularx}
\end{table}

Here $q$ counts shared fault-tree nodes and $B$ is the maximum relevant BDD size. Parameters $N,a,w,s,e$ describe the translated representations, not directly the service graph. These arithmetic bounds omit growing rational bit lengths and shared parsing/key costs; measurements include actual construction, query and process costs, memory and representation overheads.

\section{Experimental design}\label{sec:evaluation}
The experiments test the links established by the method. Independent outcomes and routing intervention assess the event--execution correspondence; calibration-to-test forecasts assess accuracy. PMX/Palladio tests the cost and accuracy of an alternative identification route, while fixed-law solver comparisons isolate computation. Full, projected and direct construction test preservation and the cost of detail. Campaigns, intervention blocks and computational repetitions remain distinct units; the supplement maps their populations and protocols.

\subsection{Applications and independent confirmation}
DeathStarBench Social Network~\cite{gan2019} supplies post composition and two timeline reads; OpenTelemetry Demo~\cite{oteldemo2026} supplies product browsing, cart insertion and checkout; Spring Petclinic Microservices~\cite{petclinic2026} supplies owner listing, owner details with visits, visit creation and veterinarian listing. Each application has two replicas of one controlled service: user-timeline, product-catalog or Visits. Other model-node eligibility is fixed true.

The confirmation design crosses three applications, two failure regimes and three repetitions, using replicas assigned to separate logical failure domains on one runner. Regime N introduces individual pauses; NC adds communication failures. Independently seeded alternating renewal schedules can overlap. Identification receives ordinary aligned observations, without fault-controller states or generator parameters. Each campaign has a 900-second calibration window and a separate 900-second test window, each containing 3,600 attempts under a balanced operation mix. The supplement specifies scheduling and duration distributions.

Whole-operation success requires all prerequisites and response checks within two seconds; timeouts and invalid required content fail. Client retries and redirects are disabled, while Petclinic retains its declared gateway retry. Visit creation also requires exactly one matching row in an independent SQL census. Raw HTTP/SQL evidence permits external-label reconstruction. Twenty-four ordinary and 24 forced 503-to-201 visit-creation controls are separate from the main series.

A separate eight-block intervention compares sham operation, settled failure of a selected replica while an alternative remains ready, and routing repair under the same outage. Each arm contains 40 checkout attempts per block. The retained historical binding audit identifies the interpretation changes subsequently tested by the independent confirmation.

All forecasts, including PMX, are fixed from calibration before test-state interpretation. Test models are constructed before opening external labels, separating realised-state compatibility with $Y$ from earlier forecast accuracy.

\subsection{Forecasts, support and accuracy}
\begin{table}[htbp]
\centering\small
\caption{Methods and targets. The first six share a recovered graph and empirical joint masked law.}\label{tab:methods}
\begin{tabularx}{\linewidth}{@{}P{0.24\linewidth}P{0.21\linewidth}Y@{}}
\toprule
Method & Target & Information and role\\\midrule
\Gstar{} & $S_uC_uT$ & Execution availability\\
\GID{} & $B_u$ & Same-law reachability\\
Selected reachability & $S_u$ & Selection alone\\
Without deadline & $S_uC_u$ & Deadline ablation\\
Without selection & $B_uC_uT$ & Selection ablation\\
Without completion & $S_uT$ & Completion ablation\\
\Gzero{} & Fixed-cardinality model & Published AINA estimator adapted to observed eligibility\\
\Bzero{} & External success frequency & Calibration labels, without smoothing\\
PMX & Conditional-local PCM & Independent extraction and Palladio calculation\\
PMX inclusive & Inclusive-failure PCM & Separate propagation convention\\\bottomrule
\end{tabularx}
\end{table}

The \Gzero{} adaptation retains AINA's uniform fixed-cardinality calculation and rounding/minimum-positive-count rule, changing both the event and state law. Same-law ablations vary event requirements while retaining the observation law; \GID{} uses within-attempt delivery eligibility. PMX receives separate native and external calibration files, including local failure labels, but neither our graph, proxy coordinates, fitted law nor test data. Fixed application mappings and known-probability controls qualify its extraction and reliability path.

Point forecasts are assessed against complete test frequencies using signed error, absolute error in percentage points and the Bernoulli Brier score~\cite{gnref2007}. Intervals and refusals retain their status. Primary contrasts compare \Gstar{} against \Gzero{} and PMX per application, using campaigns with point forecasts for every operation from both methods. For operation set $\mathcal U$, common campaigns $\mathcal C_h^{ij}$ in condition $h$ and retained conditions $\mathcal H^{ij}$,
\begin{align}
 d_c^{ij}&=\frac{100}{|\mathcal U|}\sum_{u\in\mathcal U}
 (|\widehat p_{icu}-y_{cu}|-|\widehat p_{jcu}-y_{cu}|),\\
 \widehat\Delta^{ij}&=\frac1{|\mathcal H^{ij}|}\sum_{h\in\mathcal H^{ij}}
 \frac1{|\mathcal C_h^{ij}|}\sum_{c\in\mathcal C_h^{ij}}d_c^{ij}.\label{eq:primary}
\end{align}
The protocol uses 10,000 paired whole-campaign bootstrap draws within conditions, with percentile endpoints at nominal level $1-0.05/6$ for six contrasts. The resulting ranges describe resampling variation; the small campaign counts do not establish their frequentist coverage. Any retained condition with fewer than two common campaigns prevents interval estimation. Comparisons remain conditional on point support; statistical equivalence requires a justified margin and design~\cite{lakens2017}. The supplement gives own-support and operation-level results.

\subsection{Direct identification and observation loss}\label{sec:detail-design}
The direct-target experiment reuses the eighteen confirmation campaigns to compare $E_u$ with $R_u$ and three $R_u$ routes: optimized full-model evaluation, full construction plus projection, and direct identification. Full construction must reproduce its frozen model; primitive observations, projected laws, endpoints and statuses must agree on common support. Other support and unexplained boundaries are reported separately. Interpretation excludes external labels.

Calibration-derived $R_u$ forecasts use the frozen $E_u$ forecasts' verified test counts and error aggregation, with rational equality checks. This is retrospective: source test outcomes were already available. Test-state refusal does not remove a calibration forecast with a known external test frequency.

Masking crosses 57 supported test cases, six levels (0, 10, 25, 50, 75 and 100\%), three mechanisms and five reveal settings. Both routes hide the same primitive cells before completion aggregation. Original unknowns, graph/group declarations, attempt census and deadline stay fixed. Loss is uniform, associated with an observed failed logical completion, or grouped by attempt; the final group can be split to meet the coordinate budget. External $Y$ serves only compatibility checking. The earlier $Y$-associated synthetic mask remains in the supplement.

Each setting yields joint bounds for $E_u,R_u,D_u=R_u-E_u$, widths, determinate-attempt counts and two-direction outcome compatibility. Reveals restore selection/receiver evidence, entry completion, mandatory completions or all hidden evidence, alongside no reveal. Their actual budgets are reported separately.

\subsection{Computational boundaries}\label{sec:performance-design}
All 59 calibration models enter eight configurations: on-demand and prepared evaluation, fixed-order/sifted CUDD, conditioned/joint-target aGrUM, and sparse/compact Storm, each with original and completion-projected inputs. Three fresh-process rounds rotate configuration order. Managers persist through construction, repetition, changed counts, changed masks, metadata changes, an added synchronous edge and restoration. Native inference caches are allowed; final-answer caches are disabled. Timing includes conversion, key checks, construction/updates and all eleven answers. Witnesses and independent verification follow the streams. Imports precede query timing but remain in process resources.

A 27-pair experiment measures installed tools and sealed calibration telemetry through first saved forecasts. The graph route constructs six variants plus \Gzero{} and \Bzero{}; PMX independently projects and extracts both failure variants, collects models and runs one application-only solver pass. Startup, serialization and PMX's prescribed 20-second wait per operation are included. Acquisition, installation, common role preparation and extra verification are excluded. Three split/N inputs are copied to 3,600, 7,200 and 14,400 attempts with unique identities and unchanged empirical laws. Three sequential pairs per workload alternate method order on one runner per application. Verification follows all application pairs; ratios are formed within pairs~\cite{georges2007,kalibera2013}.

The detail-cost experiment measures full $E_u$, full $R_u$, full construction plus projection to $R_u$, and direct $R_u$ at the same volumes. Warm file caches, fresh processes and three rotated rounds give 108 processes. Full $E_u$ uses exact completion reduction; full $R_u$ queries only relevant coordinates. Wall time, CPU and peak RSS cover launch through the first saved model/answer, with separate reading, interpretation, representation, counting, projection, query and saving stages. Thirty-one later ready-model queries are timed separately; verification follows all application processes. Copied records add volume, not independent accuracy observations.

\section{Experimental results}\label{sec:results}
\subsection{Correspondence to externally verified success}\label{sec:historical}
The routing intervention demonstrates the structural gap in~\eqref{eq:gap}: checkout succeeds in all 320 sham attempts, fails in all 320 settled-failure attempts despite a ready alternative replica, and succeeds in all 320 repaired-routing attempts under the same outage. Every settled failure selects the paused replica. Both prespecified block contrasts equal $100\pp$, with exact two-sided label-swap probabilities $0.0078125$; static controls succeed throughout.

The independent confirmation tests the within-attempt delivery and logical-retry bindings that resolve the historical proxy/transport interpretation errors (supplement).

Primary HTTP/SQL evidence reconstructs all 64,800 test outcomes with zero label mismatches. Models support 61,200 test attempts: none excludes the verified outcome in either direction, while 4,411 permit both success and failure. The 3,600 unsupported attempts are DeathStar user-timeline under NC in all three repetitions, a systematic support boundary. All 48 directed controls, including 24 forced retries, have independently verified success and matching determinate model outcomes.

\begin{table}[htbp]
\centering\small
\caption{Independent confirmation. P/I/R denotes point forecasts, intervals and refusals. Forecast support concerns calibration models; test-state support concerns separately interpreted attempts.}\label{tab:confirmation}
\begin{tabular}{@{}llrrrr@{}}
\toprule
Application & P/I/R & Test attempts & Supported & Ambiguous & Incompat.\\
\midrule
DeathStar & 18/0/0 & 21,600 & 18,000 & 1,382 & 0\\
OTel & 12/5/1 & 21,600 & 21,600 & 3,029 & 0\\
Petclinic & 24/0/0 & 21,600 & 21,600 & 0 & 0\\
\bottomrule
\end{tabular}

\end{table}

\subsection{Held-out accuracy and the frequency limit}\label{sec:heldout-results}
Across 60 planned calibration cases, \Gstar{} returns 54 points, five intervals and one structural refusal. The intervals occur in OTel under communication failures and have widths of 25.67--$59.83\pp$. Own-support MAEs are 0.944, 2.000 and $0.926\pp$ for DeathStar, OTel and Petclinic. Table~\ref{tab:primary} compares forecasts on complete common campaigns.

\begin{table}[htbp]
\centering\small
\caption{Held-out MAE comparisons in percentage points. $\Delta$ is \Gstar{} minus comparator; negative values favour \Gstar{}. Counts are complete common campaigns and operation cells. Brackets are descriptive campaign-bootstrap ranges. For OTel/PMX they equal the minimum and maximum of its three observed campaign contrasts.}\label{tab:primary}
\setlength{\tabcolsep}{4pt}
\begin{tabular}{@{}llrrrrl@{}}
\toprule
Application & Comparator & Camp./cells & $G^\star$ & Other & $\Delta$ & Interval\\
\midrule
DeathStar & $G_0$ & 6/18 & 0.944 & 5.954 & -5.009 & [-5.773, -3.907]\\
DeathStar & PMX & 6/18 & 0.944 & 0.944 & 0.000 & [0.000, 0.000]\\
OTel & $G_0$ & 4/12 & 2.000 & 15.657 & -13.657 & --\\
OTel & PMX & 3/9 & 1.889 & 1.971 & -0.082 & [-0.419, 0.467]\\
Petclinic & $G_0$ & 6/24 & 0.926 & 3.125 & -2.199 & [-2.431, -1.977]\\
Petclinic & PMX & 6/24 & 0.926 & 0.922 & 0.004 & [-0.003, 0.009]\\
\bottomrule
\end{tabular}

\end{table}

Observed MAE differences from PMX are approximately 0.000, $-0.082$ and $+0.004\pp$ on complete common campaigns; OTel contributes nine cells in three N campaigns. This does not bound individual differences: OTel checkout, N, repetition one gives 0.760833 for \Gstar{} versus 0.774841 for PMX, a $1.401\pp$ gap. The comparisons are descriptive on this support and do not establish statistical equivalence.

All 59 supported \Gzero{} forecasts equal one: failures are allocated only within the observed controlled-replica pool, other node eligibility is fixed, and every admitted failed-count state retains the required paths. Its casewise MAE therefore equals the external test-failure percentage. Table~\ref{tab:primary} characterizes this degenerate historical adaptation. Same-law variants permit MAE contrasts on common point support. Where \GID{} or selected reachability lose that support, bounds and support changes characterize the altered event; the routing intervention separately tests the delivery mechanism.

All 54 point forecasts equal the corresponding \Bzero{} calibration frequency exactly. Under the correspondence in~\eqref{eq:adequacy}, this is the frequency limit of~\eqref{eq:frequency-limit}. The model additionally supplies the event decomposition, requirements on observations, interval estimates and structural diagnostics. PMX/frequency agreement also follows from the composition of its calibrated root and child-success factors in inspected PCM models. A separate forty-model recomputation reproduces those saved probabilities within $1.11\times10^{-16}$; the supplement gives the factorization and examples, including PMX's use of external calibration failures.

\subsection{Necessary detail and direct identification}\label{sec:detail-results}
All 116 supported calibration/test models agree on direct primitive observations and the three $R_u$ routes; full construction reproduces the frozen model. Unmasked $E_u/R_u$ endpoints also agree. Calibration support remains 54 points, five intervals and one refusal, with rationally identical probabilities and endpoints. The retrospective $R_u$ forecasts thus preserve MAE, signed error and Brier scores exactly. Support does not expand: unexplained native-parent boundaries cause the same one calibration and three test refusals.

The joint event comparison gives additional information. Among 61,200 supported test attempts, 56,791 exclude $R_u\ne E_u$ for every compatible state; 4,409 leave it possible. For DeathStar compose-post under NC, repetition zero, both events have interval $[11/20,107/120]$, while $R_u-E_u$ has interval $[0,41/120]$. Thus equal marginal intervals can coexist with unresolved structural differences. Two OTel attempts have ambiguous success but determinate zero difference, explaining why the 4,411 ambiguous successes and 4,409 ambiguous differences have distinct meanings.

All 5,130 masking settings preserve agreement among the three $R_u$ routes. Separate $E_u/R_u$ bounds coincide in 4,227 settings; in 903, only the $E_u$ lower bound is smaller. Narrower $R_u$ intervals concern the reduced event. Table~\ref{tab:direct-missingness} pools the same 61,200 attempts. At 50\% loss, every mechanism hides 220,803 of 441,619 known native coordinates. Uniform masking leaves 1,437 attempts with known $R_u=1$ and unresolved $K_u$; failure-associated masking leaves 1,598, showing the information carried by structural evidence.

\begin{table}[htbp]
\centering\small\setlength{\tabcolsep}{4pt}
\caption{Matched primitive-observation loss: attempt-weighted identification widths and determinate counts. Deadline observations remain available. Widths concern the respective events $E_u$ and $R_u$.}\label{tab:direct-missingness}
\begin{tabular}{@{}lrrrr@{}}
\toprule
Masking & Width $E$, pp & Width $R$, pp & Point $E$ & Point $R$\\
\midrule
Baseline & 7.208 & 7.208 & 56,789 & 56,789\\
50\% uniform & 87.420 & 85.072 & 7,699 & 9,136\\
50\% failure-associated & 87.874 & 85.263 & 7,421 & 9,019\\
50\% grouped & 50.114 & 50.109 & 30,530 & 30,533\\
100\% native loss & 93.075 & 93.075 & 4,238 & 4,238\\
\bottomrule
\end{tabular}

\end{table}

Restoring selection/receiver evidence after 50\% uniform masking leaves the $R_u$ width unchanged and reduces the upper count for $D_u$ from 34,109 to 4,689. Restoring mandatory completions instead reduces $R_u$ width from 85.072 to $49.944\pp$, using a different information budget. Full restoration returns the original widths and 4,409 unresolved differences. Compatibility with external outcomes is retained throughout the coarsenings, as expected from fiber inclusion.

\subsection{Equivalent exact computation}\label{sec:scaling}
All $59\times8\times2\times7\times3=19{,}824$ measurements qualify: eleven bound pairs and statuses agree with exact references, without failures or timeouts. This is computational agreement on the same identified tasks; the preceding experiment assesses predictive accuracy against external operations.

\begin{table}[htbp]
\centering\small
\caption{Original input: median paired competitor/prepared time ratio. Each model first aggregates its three round ratios, then the application median is taken. Values above one favour the prepared evaluator. All six native configurations and both representations appear in the supplement.}\label{tab:backend-ratios}
\begin{tabular}{@{}llrrr@{}}
\toprule
Phase & Native configuration & DeathStar & OTel & Petclinic\\
\midrule
Initial & CUDD fixed & 1.99 & 2.08 & 4.50\\
Initial & aGrUM joint & 6.56 & 6.18 & 7.14\\
Initial & Storm compact & 2.66 & 2.54 & 2.77\\
Repeat & CUDD fixed & 1.09 & 1.19 & 1.14\\
Repeat & aGrUM joint & 4.84 & 5.18 & 4.32\\
Repeat & Storm compact & 1.72 & 1.04 & 1.50\\
Add edge & CUDD fixed & 0.41 & 0.43 & 0.54\\
Add edge & aGrUM joint & 6.92 & 6.01 & 8.11\\
Add edge & Storm compact & 2.25 & 1.92 & 2.32\\
\bottomrule
\end{tabular}

\end{table}

Prepared evaluation has a lower initial-answer median than every native configuration on all 59 models under both representations, with application-level paired ratios of 1.84--9.84. Original-input prepared medians are 1.781, 1.547 and 1.067 milliseconds; direct evaluation takes 1.346, 2.125 and 0.838 milliseconds. Preparation benefits some workloads but is not essential to inexpensive first evaluation.

Identical-query reuse brings fixed-order CUDD within factors of 1.09--1.19 of prepared evaluation; compact Storm has ratios of 1.04--1.72. Structural updates favour retained CUDD managers: fixed-order CUDD is 1.86--2.46 times faster when an edge is added (Table~\ref{tab:backend-ratios}). The common completion projection preserves answers and retains this distinction between initial construction and updates. Process-memory medians are 34.4--36.1 MiB for prepared evaluation, 59.6--61.4 for CUDD, 46.4--47.9 for aGrUM and 93.3--94.8 for Storm; these include imports and the complete seven-phase stream.

A separate synthetic grid exposes the residual control exponent. With eight graph nodes and ten unknown controls, direct evaluation takes 33.45 milliseconds, preparation 46--48 milliseconds and sifted CUDD 5.10 milliseconds under common projection. Complete grid results and allocation failures appear in the supplement.

\subsection{Telemetry-to-forecast cost}\label{sec:pipeline-results}
All 27 revised-pipeline pairs preserve their method-specific forecasts, bounds and statuses within $10^{-12}$. Median paired PMX/graph wall-time ratios range from 31.73 to 100.24 across the nine application/volume workloads (Table~\ref{tab:pipeline}). At 3,600 attempts, the graph route takes 1.77--3.22 seconds and PMX takes 167.91--178.68 seconds. Both timings stop at the first saved forecasts and exclude the subsequent verification work. Figure~\ref{fig:pipeline} shows the measured costs across input volumes.

\begin{figure}[htbp]\centering
\includegraphics[width=\linewidth]{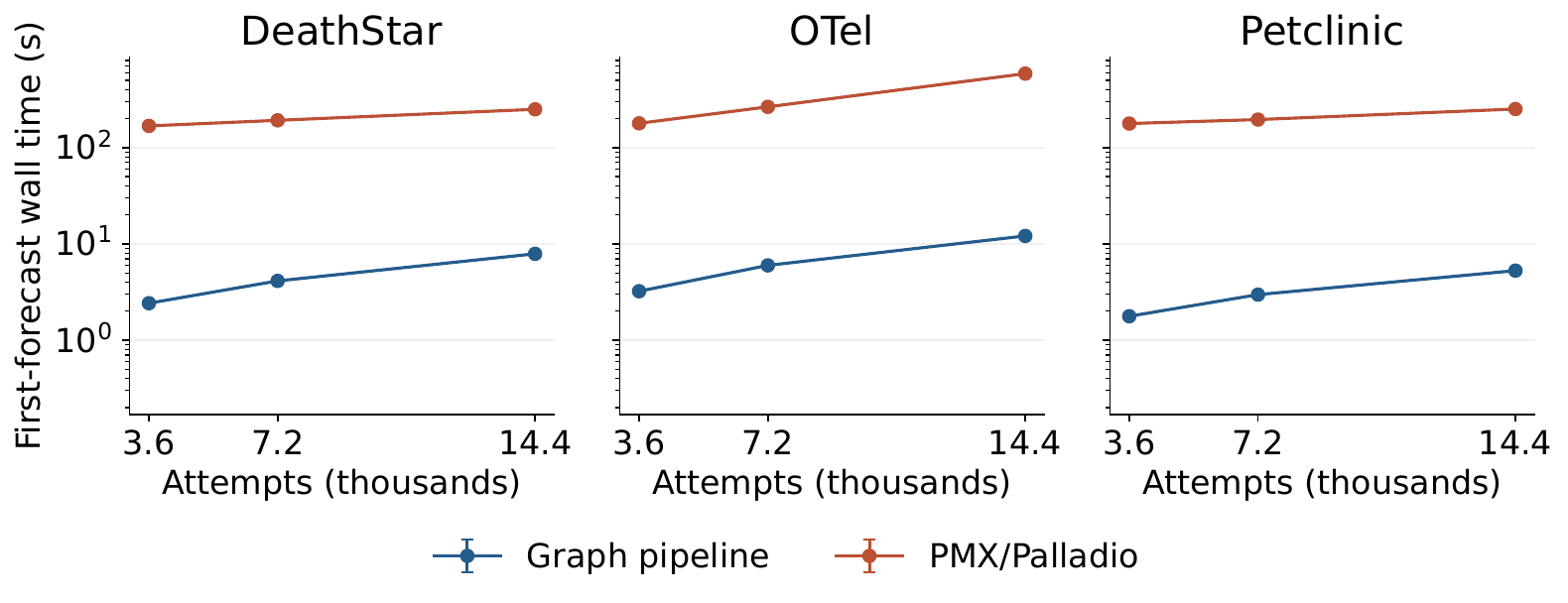}
\caption{Full first-forecast wall time for the graph route and PMX/Palladio. Points are medians and bars show minimum and maximum over three paired repetitions. The vertical axis is logarithmic. Startup and prescribed launcher waits are included. Input copies vary record count at fixed graph structure and empirical law.}\label{fig:pipeline}
\end{figure}

\begin{table}[htbp]
\centering\small
\caption{Complete reconstruction to the first saved forecasts. Times are median seconds over three pairs; ratios are medians of within-pair ratios. The measured PMX launch waits are included.}\label{tab:pipeline}
\begin{tabular}{@{}lrrrr@{}}
\toprule
Application & Attempts & Graph, s & PMX, s & PMX/graph, $\times$\\
\midrule
DeathStar & 3,600 & 2.421 & 167.907 & 69.37\\
DeathStar & 7,200 & 4.127 & 191.862 & 46.05\\
DeathStar & 14,400 & 7.885 & 249.690 & 31.73\\
OTel & 3,600 & 3.224 & 178.677 & 55.43\\
OTel & 7,200 & 5.980 & 264.828 & 44.50\\
OTel & 14,400 & 12.093 & 586.017 & 48.77\\
Petclinic & 3,600 & 1.772 & 177.718 & 100.24\\
Petclinic & 7,200 & 2.974 & 195.375 & 65.82\\
Petclinic & 14,400 & 5.276 & 251.446 & 47.82\\
\bottomrule
\end{tabular}

\end{table}

Fourfold record growth multiplies graph time by 3.26, 3.75 and 2.98, and PMX time by 1.49, 3.28 and 1.41. At 14,400 attempts, peak RSS ranges are 521--1,474 MiB and 3,799--4,380 MiB respectively; the supplement gives full resources and pairs.

The prescribed PMX waits contribute 60 or 80 seconds per pipeline. Arithmetically subtracting them leaves original-volume paired ratios of 44.58, 36.82 and 55.09; this is a sensitivity calculation, not a measured persistent implementation. This experiment uses the full graph route. Its ratios quantify the studied extraction-and-analysis lifecycles and are not a consequence of the two-coordinate direct reduction.

\subsection{The cost of necessary detail}\label{sec:direct-cost-results}
All 108 processes qualify. Direct $R_u$ is faster than each alternative in all 27 matched positions; in every one of the nine workloads, its slowest repeat takes less time than the fastest full-$E_u$ repeat. Descriptive median paired savings over three rounds per workload are 5.08--11.16\% against full $E_u$ (Table~\ref{tab:direct-cost}). At 3,600 attempts, median paired absolute savings range from 0.114 to 0.263 seconds across applications. For the same $R_u$ target, savings are 4.70--11.31\% against full construction and 4.69--11.71\% against full construction plus projection.

\begin{table}[htbp]
\centering\footnotesize\setlength{\tabcolsep}{3pt}
\caption{Matched first-answer costs in seconds. Times are medians of three processes; saved time is the median paired reduction relative to full $E_u$. Its different-target comparison is supported by the separate forecast and joint-event analyses.}\label{tab:direct-cost}
\begin{tabular}{@{}lrrrrrr@{}}
\toprule
Application & Attempts & Full $E$ & Full $R$ & Projected $R$ & Direct $R$ & Saved, \%\\
\midrule
DeathStar & 3,600 & 1.956 & 1.959 & 1.958 & 1.840 & 5.84\\
DeathStar & 7,200 & 3.317 & 3.318 & 3.318 & 3.103 & 6.44\\
DeathStar & 14,400 & 6.186 & 6.200 & 6.196 & 5.495 & 11.16\\
OTel & 3,600 & 4.377 & 4.311 & 4.322 & 4.113 & 6.02\\
OTel & 7,200 & 7.832 & 7.893 & 7.872 & 7.434 & 5.08\\
OTel & 14,400 & 15.592 & 15.464 & 15.574 & 14.738 & 5.78\\
Petclinic & 3,600 & 2.035 & 2.051 & 2.028 & 1.908 & 5.95\\
Petclinic & 7,200 & 3.366 & 3.368 & 3.357 & 3.143 & 6.31\\
Petclinic & 14,400 & 6.179 & 6.193 & 6.216 & 5.486 & 11.09\\
\bottomrule
\end{tabular}

\end{table}

The saving occurs mainly before the final query. For DeathStar at 3,600 attempts, representation construction falls from 0.173 to 0.094 seconds. Reading remains about 0.227 seconds, interpretation 0.75--0.77 seconds and startup/dispatch 0.77 seconds. Total coordinates across operations fall from 29 to 6, 38 to 6 and 25 to 8; observed categories fall from 10 to 6, 22 to 7 and 13 to 8. All routes retain the same native-span inputs, leaving peak RSS essentially unchanged (about 519, 1,475 and 643 MiB at the largest volume).

The 11,160 later ready-model queries isolate query cost: direct and projected $R_u$ share the representation and query, with operation/volume medians around four to six microseconds. Full stages, query ranges, CPU and memory appear in the supplement. The direct route's additional saving comes from avoiding construction work under the proven target specialization; it is measured separately from the PMX/Palladio ratios.

\FloatBarrier
\section{Discussion}\label{sec:discussion}
The study supports the lightweight-model hypothesis within the declared operation contracts: observed mean forecast errors are close to those of independently extracted PCM models on complete common campaigns, at substantially lower measured first-forecast cost. Exact analysis is also competitive with equivalent Boolean, Bayesian and Markov implementations, with advantages depending on the workload.

\subsection{Identification and necessary detail}
The contract defines success, bindings interpret records, and the joint observation law determines compatible probabilities. At target identifiability the estimate becomes the empirical success frequency under the event correspondence; otherwise it exposes the unresolved requirements and their attainable range. The method's benefit is therefore assessed through justified inputs, necessary detail and cost, together with accuracy. Numerical equality to a calibration frequency is consistent with that task. Collection choices affect binding soundness, model support and identification width.

Direct identification preserves the completion-and-deadline image without building the full state. The full semantics establishes the preservation conditions; the direct route realizes the resulting saving in construction. Full structure additionally supports characteristics absent from the reduced law. The 4,409 unresolved test contrasts identify possible information loss, without asserting that the events actually differed on those executions. Selection evidence narrows this contrast, whereas completion evidence narrows the reduced success range. The probability requested and the distinctions needed to interpret it determine the appropriate detail.

\subsection{Comparators and scope}
PMX/Palladio supplies an independently identified architectural/Markov model and its full process lifecycle. CUDD, aGrUM and Storm instead solve equivalent encodings of our identified problem. Their results support choosing a representation for the workload: on-demand reduction exploits few unknown controls, prepared tables amortize fixed structure, and CUDD shares subfunctions and reuses managers during structural updates. The identification method retains its meaning across those numerical choices.

The adapters cover ten operations, two replicas of one controlled service per application and fixed eligibility elsewhere. This panel does not stress-test $B_u(X)$ under failures along multiple service stages or joint $J_s(X,D)$ requirements for several controlled services. They require declared groups, stable identities and justified call boundaries; authoring these declarations and bindings is outside the measured runtime. Confirmation covers two failure regimes in one split placement; the direct-target study is retrospective. Additional physical constraints require their correct projected fibers. Sampling, regime change and semantic discrepancy enter~\eqref{eq:robust-containment} separately from identification width.

Forecast comparisons retain point-support restrictions and campaign-level units. Masking preserves group/graph declarations and deadlines; reveals have different budgets. Computational input copies vary volume at fixed structure and law. First reconstruction, ready queries and updates retain separate timings. These conditions define the scope of the measured efficiency and adequacy.

\section{Conclusion}\label{sec:conclusion}
The method identifies a lightweight stochastic microservice availability model from telemetry and operation declarations. Explicit execution semantics connects the discovered graph and observations to compatible joint laws. Formal analysis establishes attainable probabilities and when full-model projection and direct completion-and-deadline identification yield the same reduced family. A joint contrast describes the discarded structural information.

Independent confirmation supports the event correspondence on supported attempts and yields observed mean absolute errors close to PMX on complete common campaigns with point forecasts from both methods. Intervals and refusal expose identification boundaries. Equivalent solvers agree, with workload-dependent costs. The full graph route has lower measured first-forecast cost than the studied PMX/Palladio workflow. Direct identification separately preserves the retained forecasts and intervals, with descriptive median paired savings of 5.1--11.2\% against full execution analysis. These results justify required evidence and model detail, and measure their computational benefit.

\section*{Data and software availability}
The \href{https://github.com/a-a-k/telemetry-availability-identification/tree/61e65e4c09e25dd3c752756d07fe91ca504ece2f}{frozen research repository} (\texttt{simpat-artifacts-v2.1}) supplies implementation, bindings and protocols; its \href{https://github.com/a-a-k/telemetry-availability-identification/blob/61e65e4c09e25dd3c752756d07fe91ca504ece2f/docs/EXPERIMENT_MAP.md}{experiment map} links results to original evidence. The numerical package reconstructs tables, figures and finite checks. Three separately inventoried native archives supply the traces and serialized models required for raw-observation replay. The assigned evidence/software deposit identifier is \href{https://doi.org/10.5281/zenodo.22734273}{10.5281/zenodo.22734273}.

\section*{Declaration of generative AI and AI-assisted technologies in manuscript preparation}
OpenAI ChatGPT/Codex assisted with English-language drafting, literature organization, numerical cross-checks and preparation of LaTeX and plotting code. Reported experimental measurements originate from the retained research artifacts; generated text was not used as experimental data. Responsibility for the scientific content and conclusions rests with the authors.

\setlength{\bibsep}{4pt plus 1pt minus 1pt}
\bibliographystyle{elsarticle-num}
\bibliography{references}
\clearpage
\renewcommand{\normalsize}{\fontsize{10}{12}\selectfont}
\renewcommand{\small}{\fontsize{9}{11}\selectfont}
\renewcommand{\footnotesize}{\fontsize{8}{9.5}\selectfont}
\normalsize
\section*{Supplementary material}
\setcounter{section}{0}
\setcounter{subsection}{0}
\setcounter{table}{0}
\setcounter{figure}{0}
\setcounter{equation}{0}
\renewcommand{\thesection}{S\arabic{section}}
\renewcommand{\thetable}{S\arabic{table}}
\renewcommand{\thefigure}{S\arabic{figure}}
\renewcommand{\theequation}{S\arabic{equation}}
\renewcommand{\theHsection}{supplement.\arabic{section}}
\renewcommand{\theHtable}{supplement.\arabic{table}}
\renewcommand{\theHfigure}{supplement.\arabic{figure}}
\renewcommand{\theHequation}{supplement.\arabic{equation}}
\setlength{\LTpre}{8pt}
\setlength{\LTpost}{8pt}

\section{Experimental evidence map}
\begin{table}[htbp]
\centering\small
\caption{Evidence sets and their roles. Calibration and test windows each contain 3,600 attempts per campaign. Additional masking and timing reuse retained observations.}\label{tab:studies}
\begin{tabularx}{\linewidth}{@{}P{0.23\linewidth}P{0.25\linewidth}Y@{}}
\toprule
Study & Population & Purpose\\\midrule
Routing intervention & Eight independent paired blocks; three arms & Test structural reachability against actual selected delivery\\
Independent confirmation & 18 campaigns; 64,800 test attempts; 48 separate controls & Validate revised bindings and frozen forecasts\\
Direct identification and information loss & 18 saved campaigns; 57 supported test cases & Exact target projection, forecast preservation and 5,130 masking settings\\
Equivalent solvers & All 59 constructed calibration models & Eight variants, two representations, seven phases, three rounds\\
Complete reconstruction & Three fixed application inputs; three volumes; three pairs each & Telemetry-to-first-forecast cost against PMX/Palladio\\
Detail and construction cost & The same three input volumes; four routes; three rounds & 108 processes comparing full and direct identification\\\bottomrule
\end{tabularx}
\end{table}

\section{Operation contracts and acquisition}\label{app:contracts}
\begin{table}[htbp]
\centering\small
\caption{External operation contracts. A common two-second deadline covers the entire listed sequence. The evaluator checks the response requirements independently of the model's native completion observations.}\label{tab:operations}
\begin{tabularx}{\linewidth}{@{}P{0.24\linewidth}YY@{}}
\toprule
Application and operation & Request sequence & Required external result\\
\midrule
DeathStar: compose post & Submit post & Successful acknowledgement for the declared text/media payload\\
DeathStar: home timeline & Read home timeline & Accepted list representation; required fields in every returned post\\
DeathStar: user timeline & Read user timeline & Same post-field contract for the declared user\\
OTel: browse product & Read product & Exact product identity and required description, categories and price fields\\
OTel: add to cart & Read product; add item & Valid prerequisite and exact user, product and quantity in cart\\
OTel: checkout & Read product; add item; checkout & Valid prerequisites and complete order, tracking, address, shipping and item fields\\
Petclinic: list owners & Read owners & Complete declared ten-owner, thirteen-pet fixture\\
Petclinic: owner details & Read owner with visits & Complete owner/pets and all four required visits; empty fallback fails\\
Petclinic: create visit & Create visit & Correct timely response and exactly one matching row in a separate SQL census\\
Petclinic: list vets & Read veterinarians & Complete six-veterinarian fixture and specialties\\
\bottomrule
\end{tabularx}
\end{table}

The deadline spans the first HTTP request through checking the last required response; prerequisites share one attempt identity and time budget. Timeouts, invalid intermediate responses and missing required content fail. Client retries and redirects are disabled; application behaviour remains within the declared scope. Absence of a native error flag does not establish external success.

The contracts have explicit verification boundaries. A DeathStar acknowledgement is not an independent durable-storage certificate, and its timeline checks do not establish that no post was omitted. OTel checkout validates the returned order, not an actual payment or eventual delivery. Petclinic's separate SQL census establishes the presence of a matching record when audited; it does not recover its exact commit time or turn a late response into a timely success.

Each fault factor has an independently seeded alternating renewal schedule: a fixed minimum plus exponential up duration and a uniform down duration. Individual, communication and domain factors respectively use minimum/mean up durations of 12/35, 20/55 and 30/75 seconds, and down intervals of [5,9], [4,8] and [6,10] seconds. Factors can overlap, and shared domain membership induces replica-state dependence. Identification uses aligned observations, not generator distributions or hidden controller states.

Every campaign schedules 60 seconds of baseline, 900 seconds of calibration and 900 seconds of testing, with 70-second recovery intervals and a final trace flush. Four whole-operation attempts per second follow a fixed balanced operation mix, with 16 workers. This gives 3,600 calibration and 3,600 test attempts per campaign, and 64,800 test attempts across the 18 confirmation campaigns. Traces and proxy checks are associated with attempt identities. Revised delivery eligibility uses the native entry and operation time window described in the main article; earlier proxy observations remain contextual. The single-runner deployment shares a host time base. For multi-host inputs, delivery requires a justified timestamp-error bound that leaves window membership unambiguous; unresolved membership remains unknown. The revised collector records probe start and end and retains raw SQL query rows for external verification.

PMX receives separate native and external calibration files, including failure labels, but neither our graph, proxy states or fitted parameters nor any test observations. DeathStar is projected at service boundaries; OTel uses explicit SERVER spans; Petclinic adds observed MySQL CLIENT peers to SERVER spans. OTel database execution remains within its enclosing SERVER abstraction. Caller-context separation, propagation conventions and integer probability-mass bridges are fixed adaptations. The methods share attempt populations and target outcomes, with method-specific observations.

The primary PMX variant distinguishes local failures from propagated failures and can refuse a point when its conditional quantities lack support. The inclusive variant retains a different composition convention and is reported separately. Four known-probability context/database models, two variants and two solver passes provide 16 control results per application batch. The adopted PMX build includes a narrowly scoped clock-progress correction for a zero-width numerical step, qualified before main data acquisition on identical development inputs. It preserves all previously supported forecasts and the original observations and interval endpoints. The binary and its source are pinned in the artifact.

\section{Binding validation and evidence populations}\label{app:provenance}
Eligibility denotes delivery within the attempt window. A replica-linked native SERVER or declared source entry establishes delivery even for an erroneous response; missing evidence remains unknown. Historical proxy values and masks remain joint contextual coordinates, marginalized without independence assumptions. Petclinic's declared single POST retry after 503 in the gateway--Visits group requires sequential client attempts with a common parent, explicit final success and one identified successful receiver. Ambiguous ordering, receiver identity or possible first effects leave logical completion unknown.

The external oracle checks raw HTTP requirements and independent SQL rows for visit creation. Model construction receives native traces, attempt boundaries and declarations, excluding business labels, response payloads, persistence outcomes and fault schedules. Test-state interpretation fixes graph/attempt states and saves extrema before the external evaluator opens those fields. Calibration forecasts and later native-state compatibility checks thus remain distinct.

\subsection{Historical diagnostic}
The 240-campaign historical study contains 864,000 calibration attempts and the same number of test attempts. Initial calibration construction supports 817,200 attempts; the remaining 46,800 have structural refusals. There are 2,541 incompatible recorded successes: 556 in DeathStar, 432 in OTel and 1,553 in Petclinic. None of the opposite-direction contradictions is found. Of these, 2,503 have fully observed model states and 38 have partial states. An earlier negative selected-replica proxy coordinate appears in 2,533, a logical retry in 11, and both in three.

The external calibration frequency exceeds the initial interval upper endpoint in 97 of 120 interval cases and exceeds the point estimate in 199 of 641 point cases. Aggregate agreement alone is consequently insufficient to audit the binding: opposing attempt errors could cancel. Native/probe timing in inspected examples establishes recovery before or during the request. The old sampler timestamp preceded inspection/statistics retrieval and was not an independently retained probe-completion timestamp.

An independent raw-outcome audit covers 43,200 fixed historical attempts. It reproduces 39,831 outcomes without label discrepancy; 3,369 Petclinic write outcomes lack the original SQL SELECT rows and cannot be independently re-established from acknowledgements alone. Thirty-three of the 216 initially contradictory attempts in this subset are subject to that persistence-evidence limitation; the other 183 have independently verified external success. The full historical counts above therefore concern registered outcomes, with this distinction retained.

A prespecified diagnostic subset contains 29 historical campaigns: four DeathStar, four OTel and 21 Petclinic campaigns. Of 104,400 attempts, 103,200 have supported models. Revised binding changes the contradiction count from 535 to zero in both directions. Its 107 calibration models contain 101 points and six intervals. The diagnostic exact-calculation checks cover 1,712 queries (107 models, eight variants and two representations), with 2,916 witness endpoints checked. Flipping all external labels does not change the constructed graph models. These are retrospective checks of a correction selected through the audit, not a substitute for independent confirmation.

\subsection{Independent confirmation and amendments}
Before acquisition, confirmation fixed split placement, regimes N/NC and three repetitions per application: 64,800 calibration and 64,800 test attempts, plus 48 separate controls. A proxy-template matching defect invalidated the first six Petclinic captures. Only these were recollected after correction; the twelve sealed DeathStar/OTel captures were retained. Replacement selection was independent of outcomes.

All eighteen final forecast bundles were fixed by 16:49:12 UTC on 12 September 2026. The first test-state interpretation began at 16:49:14; the last finished by 16:50:19; the first independent external-outcome evaluator began at 16:50:21. All PMX numerical forecasts are included in the first boundary. All 64,800 test labels match independently reconstructed primary outcomes. The 57 supported test-state cases contain 61,200 attempts, 4,411 ambiguous outcomes and zero incompatible outcomes. All 1,960 retained extremal witness endpoints pass verification. The 24 ordinary and 24 forced retry controls all succeed independently and have determinate matching model outcomes; two further declared retry groups occur naturally in the main Petclinic observations.

A summary-generation repair marked colocated-to-split transfer as outside this deliberately split-only design. It changed no observations, calibration forecasts, model bounds, forecast errors or bootstrap results. The historical departure labelled C12 remains separate: the initial 240-campaign study reconstructed numerical PMX forecasts from previously fixed PCM models and solver after test outcomes had been opened. That historical study is used here for diagnostic evidence; it is not represented as having the confirmation series' forecast-freezing order.

\subsection{Concrete PMX/frequency agreement}
The independent PMX audit uses 40 saved PCM models (20 operation cases under two propagation variants) from six fixed historical campaigns. Re-evaluation differs from saved solver values by at most $1.11\times10^{-16}$. In the inspected local-failure convention, the root failure probability is estimated from external calibration outcomes conditional on no observed immediate-child failure. Let $e_0$ be failures among those $n_0$ attempts, out of $N$ total attempts. Then
\[
 B_0=(1-e_0/n_0)n_0/N,\qquad
 p_{\rm PMX}=(1-e_0/n_0)S,
\]
where $S$ is the actual PCM-composed child-success factor. Their difference is $(1-e_0/n_0)(S-n_0/N)$. This identity requires the stated propagation convention and $n_0>0$; an unsupported conditional denominator is not silently replaced by zero.

One DeathStar model has root success factor $1057/1127$ and child factor $1127/1200$, giving $1057/1200$, exactly the empirical frequency. In one OTel split/ND case, the factors are $(744/745)\times0.83\times(1199/1200)\times(19/24)$, giving approximately 0.655655 rather than the empirical 0.620000. This is the composition implemented by that saved model, including its loop/child treatment. It does not establish a limitation of every PCM model or Markov analysis. Across the historical comparison, PMX and frequency values agree to $10^{-12}$ in 240/240 DeathStar, 53/100 OTel and 280/320 Petclinic common point cases.

\section{Forecast support and complete numerical summaries}\label{app:forecasts}
Primary MAE contrasts use complete common campaigns, equal operation weights within campaigns, and equal retained-condition weights. Bootstrap resampling preserves whole campaigns within each condition. Seeds 771602--771607 correspond, in order, to DeathStar/$G_0$, DeathStar/PMX, OTel/$G_0$, OTel/PMX, Petclinic/$G_0$ and Petclinic/PMX. Every draw uses the same support accounting. The OTel/$G_0$ descriptive result uses four complete campaigns and twelve cells; its NC condition contains only one campaign. The OTel/PMX comparison uses three complete N campaigns and nine cells. The other primary comparisons retain all six campaigns of their application.

For OTel/PMX, resampling three campaigns with replacement admits only $3^3=27$ ordered selections. The 10,000 draws approximate this discrete distribution. At the protocol's percentile level, the reported range $[-0.419,0.467]\pp$ equals the observed minimum and maximum campaign contrasts; it supplies no calibrated frequentist coverage. DeathStar/PMX's zero range records equality on the observed campaigns, without establishing zero population uncertainty. All 59 supported $G_0$ forecasts equal one; this adapter's casewise MAE therefore equals the external test-failure percentage. These contrasts retain the protocol's historical baseline without treating it as an informative structural predictor on this panel.

The tables use each method's own point support. Dashes denote absence; intervals are not replaced by midpoints. The complete JSON retains every planned method/case, absence reasons, stable-period secondary results and split-only transfer status.

Operation codes are CP (compose post), HT (home timeline), UT (user timeline), BP (browse product), AC (add to cart), CO (checkout), LO (list owners), OV (owner with visits), CV (create visit), and LV (list veterinarians). The completion, selection and deadline ablations remove the corresponding predicate while retaining the same empirical law. PMX-inc uses its separate inclusive-failure convention.
\begin{table}[htbp]\centering\small
\caption{All nonpoint calibration forecasts from the revised execution model. Endpoints refer to the empirical calibration law; no midpoint forecast is substituted.}\label{tab:s-intervals}
\begin{tabular}{@{}llrlrrr@{}}
\toprule
Application & Regime & Rep. & Op. & Lower & Upper & Width, pp\\
\midrule
OTel & NC & 1 & AC & 0.57167 & 0.85083 & 27.92\\
OTel & NC & 1 & BP & 0.54250 & 0.83417 & 29.17\\
OTel & NC & 1 & CO & 0.47750 & 0.73417 & 25.67\\
OTel & NC & 2 & AC & 0.25417 & 0.85250 & 59.83\\
OTel & NC & 2 & BP & 0.26417 & 0.85833 & 59.42\\
OTel & NC & 2 & CO & -- & -- & Refusal\\
\bottomrule
\end{tabular}

\end{table}
\begingroup\small\setlength{\tabcolsep}{4pt}
\begin{longtable}{@{}llrrrrr@{}}
\caption{Own point-support errors. Each method retains its own supported population; absent point estimates are not imputed.}\label{tab:s-own}\\
\toprule
Application & Method & Points & MAE, pp & Signed, pp & Brier & 90th, pp\\
\midrule
\endfirsthead
\toprule
Application & Method & Points & MAE, pp & Signed, pp & Brier & 90th, pp\\
\midrule
\endhead
\bottomrule
\endfoot
DeathStar & $G^\star$ & 18/18 & 0.944 & 0.148 & 0.05361 & 2.642\\
DeathStar & $G_{\rm ID}$ & 6/18 & 0.000 & 0.000 & 0.00000 & 0.000\\
DeathStar & $S$ & 6/18 & 0.000 & 0.000 & 0.00000 & 0.000\\
DeathStar & $G_{-T}$ & 6/18 & 0.000 & 0.000 & 0.00000 & 0.000\\
DeathStar & $G_{-S}$ & 18/18 & 0.944 & 0.148 & 0.05361 & 2.642\\
DeathStar & $G_{-C}$ & 9/18 & 0.194 & -0.194 & 0.01663 & 0.750\\
DeathStar & $G_0$ & 18/18 & 5.954 & 5.954 & 0.05954 & 13.008\\
DeathStar & $B_0$ & 18/18 & 0.944 & 0.148 & 0.05361 & 2.642\\
DeathStar & PMX & 18/18 & 0.944 & 0.148 & 0.05361 & 2.642\\
DeathStar & PMX-inc & 18/18 & 0.809 & -0.186 & 0.05356 & 2.640\\
OTel & $G^\star$ & 12/18 & 2.000 & 0.491 & 0.12927 & 3.292\\
OTel & $G_{\rm ID}$ & 0/18 & -- & -- & -- & --\\
OTel & $S$ & 0/18 & -- & -- & -- & --\\
OTel & $G_{-T}$ & 10/18 & 2.042 & 0.690 & 0.13079 & 2.967\\
OTel & $G_{-S}$ & 12/18 & 2.000 & 0.491 & 0.12927 & 3.292\\
OTel & $G_{-C}$ & 12/18 & 2.000 & 0.491 & 0.12927 & 3.292\\
OTel & $G_0$ & 17/18 & 15.824 & 15.824 & 0.15824 & 24.850\\
OTel & $B_0$ & 18/18 & 2.185 & -0.213 & 0.12887 & 4.150\\
OTel & PMX & 9/18 & 1.971 & 1.930 & 0.12569 & 4.815\\
OTel & PMX-inc & 18/18 & 14.000 & -14.000 & 0.15172 & 22.302\\
Petclinic & $G^\star$ & 24/24 & 0.926 & 0.269 & 0.02931 & 2.344\\
Petclinic & $G_{\rm ID}$ & 12/24 & 0.000 & 0.000 & 0.00000 & 0.000\\
Petclinic & $S$ & 12/24 & 0.000 & 0.000 & 0.00000 & 0.000\\
Petclinic & $G_{-T}$ & 12/24 & 0.000 & 0.000 & 0.00000 & 0.000\\
Petclinic & $G_{-S}$ & 24/24 & 0.926 & 0.269 & 0.02931 & 2.344\\
Petclinic & $G_{-C}$ & 12/24 & 0.000 & 0.000 & 0.00000 & 0.000\\
Petclinic & $G_0$ & 24/24 & 3.125 & 3.125 & 0.03125 & 7.611\\
Petclinic & $B_0$ & 24/24 & 0.926 & 0.269 & 0.02931 & 2.344\\
Petclinic & PMX & 24/24 & 0.922 & 0.275 & 0.02931 & 2.344\\
Petclinic & PMX-inc & 24/24 & 1.265 & -0.591 & 0.02975 & 3.132\\
\end{longtable}

\begin{longtable}{@{}lllrrrrr@{}}
\caption{Operation-level summaries on each method's own point support.}\label{tab:s-operations}\\
\toprule
Application & Op. & Method & Camp. & Forecast & Observed & MAE, pp & Signed, pp\\
\midrule
\endfirsthead
\toprule
Application & Op. & Method & Camp. & Forecast & Observed & MAE, pp & Signed, pp\\
\midrule
\endhead
\bottomrule
\endfoot
DeathStar & CP & $G^\star$ & 6 & 0.9281 & 0.9238 & 1.236 & 0.431\\
DeathStar & HT & $G^\star$ & 6 & 1.0000 & 1.0000 & 0.000 & 0.000\\
DeathStar & UT & $G^\star$ & 6 & 0.8978 & 0.8976 & 1.597 & 0.014\\
DeathStar & CP & $G_{\rm ID}$ & 0 & -- & -- & -- & --\\
DeathStar & HT & $G_{\rm ID}$ & 6 & 1.0000 & 1.0000 & 0.000 & 0.000\\
DeathStar & UT & $G_{\rm ID}$ & 0 & -- & -- & -- & --\\
DeathStar & CP & $S$ & 0 & -- & -- & -- & --\\
DeathStar & HT & $S$ & 6 & 1.0000 & 1.0000 & 0.000 & 0.000\\
DeathStar & UT & $S$ & 0 & -- & -- & -- & --\\
DeathStar & CP & $G_{-T}$ & 0 & -- & -- & -- & --\\
DeathStar & HT & $G_{-T}$ & 6 & 1.0000 & 1.0000 & 0.000 & 0.000\\
DeathStar & UT & $G_{-T}$ & 0 & -- & -- & -- & --\\
DeathStar & CP & $G_{-S}$ & 6 & 0.9281 & 0.9238 & 1.236 & 0.431\\
DeathStar & HT & $G_{-S}$ & 6 & 1.0000 & 1.0000 & 0.000 & 0.000\\
DeathStar & UT & $G_{-S}$ & 6 & 0.8978 & 0.8976 & 1.597 & 0.014\\
DeathStar & CP & $G_{-C}$ & 0 & -- & -- & -- & --\\
DeathStar & HT & $G_{-C}$ & 6 & 1.0000 & 1.0000 & 0.000 & 0.000\\
DeathStar & UT & $G_{-C}$ & 3 & 0.9206 & 0.9283 & 0.778 & -0.778\\
DeathStar & CP & $G_0$ & 6 & 1.0000 & 0.9238 & 7.625 & 7.625\\
DeathStar & HT & $G_0$ & 6 & 1.0000 & 1.0000 & 0.000 & 0.000\\
DeathStar & UT & $G_0$ & 6 & 1.0000 & 0.8976 & 10.236 & 10.236\\
DeathStar & CP & $B_0$ & 6 & 0.9281 & 0.9238 & 1.236 & 0.431\\
DeathStar & HT & $B_0$ & 6 & 1.0000 & 1.0000 & 0.000 & 0.000\\
DeathStar & UT & $B_0$ & 6 & 0.8978 & 0.8976 & 1.597 & 0.014\\
DeathStar & CP & PMX & 6 & 0.9281 & 0.9238 & 1.236 & 0.431\\
DeathStar & HT & PMX & 6 & 1.0000 & 1.0000 & 0.000 & 0.000\\
DeathStar & UT & PMX & 6 & 0.8978 & 0.8976 & 1.597 & 0.014\\
DeathStar & CP & PMX-inc & 6 & 0.9238 & 0.9238 & 1.111 & 0.001\\
DeathStar & HT & PMX-inc & 6 & 1.0000 & 1.0000 & 0.000 & 0.000\\
DeathStar & UT & PMX-inc & 6 & 0.8921 & 0.8976 & 1.317 & -0.557\\
OTel & AC & $G^\star$ & 4 & 0.8892 & 0.8894 & 1.056 & -0.028\\
OTel & BP & $G^\star$ & 4 & 0.8947 & 0.8764 & 1.889 & 1.833\\
OTel & CO & $G^\star$ & 4 & 0.7611 & 0.7644 & 3.056 & -0.333\\
OTel & AC & $G_{\rm ID}$ & 0 & -- & -- & -- & --\\
OTel & BP & $G_{\rm ID}$ & 0 & -- & -- & -- & --\\
OTel & CO & $G_{\rm ID}$ & 0 & -- & -- & -- & --\\
OTel & AC & $S$ & 0 & -- & -- & -- & --\\
OTel & BP & $S$ & 0 & -- & -- & -- & --\\
OTel & CO & $S$ & 0 & -- & -- & -- & --\\
OTel & AC & $G_{-T}$ & 4 & 0.8893 & 0.8894 & 1.069 & -0.014\\
OTel & BP & $G_{-T}$ & 4 & 0.8947 & 0.8764 & 1.889 & 1.833\\
OTel & CO & $G_{-T}$ & 2 & 0.7579 & 0.7554 & 3.167 & 0.250\\
OTel & AC & $G_{-S}$ & 4 & 0.8892 & 0.8894 & 1.056 & -0.028\\
OTel & BP & $G_{-S}$ & 4 & 0.8947 & 0.8764 & 1.889 & 1.833\\
OTel & CO & $G_{-S}$ & 4 & 0.7611 & 0.7644 & 3.056 & -0.333\\
OTel & AC & $G_{-C}$ & 4 & 0.8892 & 0.8894 & 1.056 & -0.028\\
OTel & BP & $G_{-C}$ & 4 & 0.8947 & 0.8764 & 1.889 & 1.833\\
OTel & CO & $G_{-C}$ & 4 & 0.7611 & 0.7644 & 3.056 & -0.333\\
OTel & AC & $G_0$ & 6 & 1.0000 & 0.8883 & 11.167 & 11.167\\
OTel & BP & $G_0$ & 6 & 1.0000 & 0.8804 & 11.958 & 11.958\\
OTel & CO & $G_0$ & 5 & 1.0000 & 0.7565 & 24.347 & 24.347\\
OTel & AC & $B_0$ & 6 & 0.8814 & 0.8883 & 1.722 & -0.694\\
OTel & BP & $B_0$ & 6 & 0.8821 & 0.8804 & 2.000 & 0.167\\
OTel & CO & $B_0$ & 6 & 0.7619 & 0.7631 & 2.833 & -0.111\\
OTel & AC & PMX & 3 & 0.9033 & 0.8931 & 1.028 & 1.028\\
OTel & BP & PMX & 3 & 0.9053 & 0.8919 & 1.444 & 1.333\\
OTel & CO & PMX & 3 & 0.7923 & 0.7581 & 3.440 & 3.429\\
OTel & AC & PMX-inc & 6 & 0.7653 & 0.8883 & 12.299 & -12.299\\
OTel & BP & PMX-inc & 6 & 0.7671 & 0.8804 & 11.330 & -11.330\\
OTel & CO & PMX-inc & 6 & 0.5793 & 0.7631 & 18.372 & -18.372\\
Petclinic & CV & $G^\star$ & 6 & 0.9430 & 0.9361 & 1.907 & 0.685\\
Petclinic & LO & $G^\star$ & 6 & 1.0000 & 1.0000 & 0.000 & 0.000\\
Petclinic & LV & $G^\star$ & 6 & 1.0000 & 1.0000 & 0.000 & 0.000\\
Petclinic & OV & $G^\star$ & 6 & 0.9428 & 0.9389 & 1.796 & 0.389\\
Petclinic & CV & $G_{\rm ID}$ & 0 & -- & -- & -- & --\\
Petclinic & LO & $G_{\rm ID}$ & 6 & 1.0000 & 1.0000 & 0.000 & 0.000\\
Petclinic & LV & $G_{\rm ID}$ & 6 & 1.0000 & 1.0000 & 0.000 & 0.000\\
Petclinic & OV & $G_{\rm ID}$ & 0 & -- & -- & -- & --\\
Petclinic & CV & $S$ & 0 & -- & -- & -- & --\\
Petclinic & LO & $S$ & 6 & 1.0000 & 1.0000 & 0.000 & 0.000\\
Petclinic & LV & $S$ & 6 & 1.0000 & 1.0000 & 0.000 & 0.000\\
Petclinic & OV & $S$ & 0 & -- & -- & -- & --\\
Petclinic & CV & $G_{-T}$ & 0 & -- & -- & -- & --\\
Petclinic & LO & $G_{-T}$ & 6 & 1.0000 & 1.0000 & 0.000 & 0.000\\
Petclinic & LV & $G_{-T}$ & 6 & 1.0000 & 1.0000 & 0.000 & 0.000\\
Petclinic & OV & $G_{-T}$ & 0 & -- & -- & -- & --\\
Petclinic & CV & $G_{-S}$ & 6 & 0.9430 & 0.9361 & 1.907 & 0.685\\
Petclinic & LO & $G_{-S}$ & 6 & 1.0000 & 1.0000 & 0.000 & 0.000\\
Petclinic & LV & $G_{-S}$ & 6 & 1.0000 & 1.0000 & 0.000 & 0.000\\
Petclinic & OV & $G_{-S}$ & 6 & 0.9428 & 0.9389 & 1.796 & 0.389\\
Petclinic & CV & $G_{-C}$ & 0 & -- & -- & -- & --\\
Petclinic & LO & $G_{-C}$ & 6 & 1.0000 & 1.0000 & 0.000 & 0.000\\
Petclinic & LV & $G_{-C}$ & 6 & 1.0000 & 1.0000 & 0.000 & 0.000\\
Petclinic & OV & $G_{-C}$ & 0 & -- & -- & -- & --\\
Petclinic & CV & $G_0$ & 6 & 1.0000 & 0.9361 & 6.389 & 6.389\\
Petclinic & LO & $G_0$ & 6 & 1.0000 & 1.0000 & 0.000 & 0.000\\
Petclinic & LV & $G_0$ & 6 & 1.0000 & 1.0000 & 0.000 & 0.000\\
Petclinic & OV & $G_0$ & 6 & 1.0000 & 0.9389 & 6.111 & 6.111\\
Petclinic & CV & $B_0$ & 6 & 0.9430 & 0.9361 & 1.907 & 0.685\\
Petclinic & LO & $B_0$ & 6 & 1.0000 & 1.0000 & 0.000 & 0.000\\
Petclinic & LV & $B_0$ & 6 & 1.0000 & 1.0000 & 0.000 & 0.000\\
Petclinic & OV & $B_0$ & 6 & 0.9428 & 0.9389 & 1.796 & 0.389\\
Petclinic & CV & PMX & 6 & 0.9432 & 0.9361 & 1.892 & 0.710\\
Petclinic & LO & PMX & 6 & 1.0000 & 1.0000 & 0.000 & 0.000\\
Petclinic & LV & PMX & 6 & 1.0000 & 1.0000 & 0.000 & 0.000\\
Petclinic & OV & PMX & 6 & 0.9428 & 0.9389 & 1.796 & 0.389\\
Petclinic & CV & PMX-inc & 6 & 0.9086 & 0.9361 & 3.265 & -2.755\\
Petclinic & LO & PMX-inc & 6 & 1.0000 & 1.0000 & 0.000 & 0.000\\
Petclinic & LV & PMX-inc & 6 & 1.0000 & 1.0000 & 0.000 & 0.000\\
Petclinic & OV & PMX-inc & 6 & 0.9428 & 0.9389 & 1.796 & 0.389\\
\end{longtable}

\endgroup

\section{Incomplete-observation experiment}\label{app:missingness}
All 57 supported, independently verified test cases enter without selection by $E=Y$. Original masks, graph structure and attempt/deadline boundaries stay fixed. Additional masking hides known native coordinates without repeating graph discovery on thinned traces. The synthetic failure-associated mechanism can use $Y$; estimation and reveal choices cannot.

Each mechanism hides the same coordinate budget. Whole-attempt masking removes complete native records, splitting the final record if needed; original unknowns are never imputed. The exact upper/lower count difference equals the ambiguous-attempt count, so pooled width is the ambiguous proportion. Coarsening retains each original true compatible state; continued compatibility checks the implementation after independent unmasked validation.

\begin{table}[htbp]\centering\small
\caption{Attempt-weighted widths in percentage points on unchanged supported populations.}\label{tab:s-missingness}
\begin{tabular}{@{}llrrr@{}}
\toprule
Added loss & Mechanism & DeathStar & OTel & Petclinic\\
\midrule
0\% & Uniform coordinates & 7.68 & 14.02 & 0.00\\
0\% & Failure-dependent & 7.68 & 14.02 & 0.00\\
0\% & Whole-attempt records & 7.68 & 14.02 & 0.00\\
10\% & Uniform coordinates & 50.68 & 58.88 & 36.00\\
10\% & Failure-dependent & 48.08 & 52.71 & 35.38\\
10\% & Whole-attempt records & 16.53 & 21.21 & 9.90\\
25\% & Uniform coordinates & 75.03 & 78.97 & 66.50\\
25\% & Failure-dependent & 73.65 & 75.56 & 66.88\\
25\% & Whole-attempt records & 29.68 & 31.81 & 24.76\\
50\% & Uniform coordinates & 90.11 & 84.25 & 88.34\\
50\% & Failure-dependent & 90.06 & 84.52 & 89.29\\
50\% & Whole-attempt records & 51.56 & 49.52 & 49.50\\
75\% & Uniform coordinates & 94.95 & 84.72 & 96.48\\
75\% & Failure-dependent & 94.97 & 85.00 & 97.06\\
75\% & Whole-attempt records & 73.50 & 67.44 & 74.25\\
100\% & Uniform coordinates & 95.63 & 85.00 & 99.02\\
100\% & Failure-dependent & 95.63 & 85.00 & 99.02\\
100\% & Whole-attempt records & 95.63 & 85.00 & 99.02\\
\bottomrule
\end{tabular}

\end{table}

\begin{table}[htbp]\centering\small
\caption{Revealed coordinates / resolved ambiguous attempts after 50\% failure-dependent loss. Different rows have different actual information budgets.}\label{tab:s-reveal}
\begin{tabular}{@{}llll@{}}
\toprule
Revealed group & DeathStar & OTel & Petclinic\\
\midrule
Selection / receiver & 13,379 / 441 & 27,671 / 496 & 15,202 / 745\\
Entry outcome & 9,179 / 961 & 11,590 / 219 & 10,731 / 2,352\\
Mandatory completions & 49,162 / 3,247 & 59,910 / 1,683 & 23,979 / 4,301\\
All newly hidden & 71,720 / 14,828 & 99,171 / 15,227 & 49,912 / 19,286\\
\bottomrule
\end{tabular}

\end{table}

\begin{samepage}
The fixed semantic reveal groups do not compare optimized and fixed/random acquisition at equal instrumentation cost. Restoring all newly hidden coordinates recovers the original bounds exactly, leaving original unknowns unresolved.
\par\end{samepage}

\section{Equivalent-backend experiment}\label{app:backends}
The revised application panel contains eighteen DeathStar, seventeen OTel and twenty-four Petclinic calibration models. One additional planned OTel case has a structural refusal. All constructed models, including ambiguous ones, enter every solver configuration and both representations. There are 19,824 timed answers, 11,564 case-level paired ratios and 2,832 process resource records.

The implementations use CUDD 3.0.0 through dd 0.6.0, pyagrum 3.1.1 and stormpy 1.14.0. CUDD configurations use fixed variable order or sifting on semantic rebuilds; both retain the manager across updates. aGrUM uses LazyPropagation, either conditioned separately on observation categories or with joint category/target factors. Its auxiliary probabilities serve only to test support, and the final empirical weighting uses exact rational counts. Extra answer-verification inference is outside the timer. Storm uses a sparse decision graph or a compact exact encoding merging equivalent continuations and bypassing deterministic steps. Both represent the model's observation-compatible completions, rather than fitted physical failure/repair dynamics.

Each case/round/configuration starts one process and seven snapshots: initial model, identical query, changed counts, changed masks, metadata-only change, added synchronous edge and restoration. Backend containers persist; native inference caches are permitted, while adapter result caches are disabled. Every variant receives the same original or completion-projected input and computes eleven exact bound pairs and statuses. Timings include projection, key validation, construction, updates and queries. Witness generation and independent checks follow all timed streams for the application.

Each stream has one computational thread, a 60-second limit and 3 GiB address-space budget; aGrUM inference has a 1 GiB budget. No failure or timeout occurs in this revised panel. DeathStar and Petclinic use AMD EPYC 7763 processors; OTel uses AMD EPYC 9V74. The solver processes use Python 3.11.16. Complete environment manifests, including OS and package metadata, accompany the inputs. Ratios pair methods on the same case and round before taking medians, so they are not ratios of separately pooled median times.

\begingroup\small\setlength{\tabcolsep}{5pt}
\begin{longtable}{@{}lllrrr@{}}
\caption{All median paired native/prepared time ratios; values above one favour preparation. Each case uses the median of three round ratios.}\label{tab:s-backends}\\
\toprule
Input & Phase & Native configuration & DeathStar & OTel & Petclinic\\
\midrule
\endfirsthead
\toprule
Input & Phase & Native configuration & DeathStar & OTel & Petclinic\\
\midrule
\endhead
\bottomrule
\endfoot
Original & Initial & CUDD fixed & 1.992 & 2.084 & 4.503\\
Original & Initial & CUDD sift & 2.272 & 2.192 & 4.212\\
Original & Initial & aGrUM cond. & 8.283 & 9.843 & 8.180\\
Original & Initial & aGrUM joint & 6.562 & 6.183 & 7.139\\
Original & Initial & Storm sparse & 2.975 & 3.447 & 2.962\\
Original & Initial & Storm compact & 2.664 & 2.540 & 2.775\\
Original & Repeat & CUDD fixed & 1.089 & 1.191 & 1.143\\
Original & Repeat & CUDD sift & 1.116 & 1.229 & 1.101\\
Original & Repeat & aGrUM cond. & 17.062 & 18.635 & 12.243\\
Original & Repeat & aGrUM joint & 4.839 & 5.184 & 4.321\\
Original & Repeat & Storm sparse & 3.776 & 2.994 & 2.859\\
Original & Repeat & Storm compact & 1.722 & 1.042 & 1.503\\
Original & Counts & CUDD fixed & 1.114 & 1.163 & 1.047\\
Original & Counts & CUDD sift & 1.130 & 1.199 & 1.053\\
Original & Counts & aGrUM cond. & 17.887 & 19.121 & 11.991\\
Original & Counts & aGrUM joint & 5.143 & 5.429 & 4.367\\
Original & Counts & Storm sparse & 4.135 & 3.382 & 3.073\\
Original & Counts & Storm compact & 2.034 & 1.394 & 1.607\\
Original & Masks & CUDD fixed & 1.188 & 1.231 & 1.097\\
Original & Masks & CUDD sift & 1.180 & 1.236 & 1.075\\
Original & Masks & aGrUM cond. & 23.149 & 24.390 & 16.820\\
Original & Masks & aGrUM joint & 13.030 & 13.900 & 12.755\\
Original & Masks & Storm sparse & 11.431 & 9.792 & 8.246\\
Original & Masks & Storm compact & 10.295 & 7.659 & 7.727\\
Original & Metadata & CUDD fixed & 1.159 & 1.175 & 1.211\\
Original & Metadata & CUDD sift & 1.169 & 1.222 & 1.189\\
Original & Metadata & aGrUM cond. & 18.072 & 19.919 & 12.714\\
Original & Metadata & aGrUM joint & 5.222 & 5.444 & 4.531\\
Original & Metadata & Storm sparse & 4.953 & 3.838 & 3.449\\
Original & Metadata & Storm compact & 2.288 & 1.143 & 2.102\\
Original & Add edge & CUDD fixed & 0.406 & 0.429 & 0.537\\
Original & Add edge & CUDD sift & 0.548 & 0.589 & 0.795\\
Original & Add edge & aGrUM cond. & 9.494 & 10.509 & 9.120\\
Original & Add edge & aGrUM joint & 6.923 & 6.012 & 8.114\\
Original & Add edge & Storm sparse & 2.553 & 2.651 & 2.206\\
Original & Add edge & Storm compact & 2.251 & 1.920 & 2.324\\
Original & Restore & CUDD fixed & 0.392 & 0.388 & 0.552\\
Original & Restore & CUDD sift & 0.552 & 0.561 & 0.806\\
Original & Restore & aGrUM cond. & 10.461 & 9.851 & 9.984\\
Original & Restore & aGrUM joint & 6.763 & 6.172 & 8.857\\
Original & Restore & Storm sparse & 2.156 & 2.127 & 2.099\\
Original & Restore & Storm compact & 1.819 & 1.790 & 1.941\\
Projected & Initial & CUDD fixed & 2.194 & 1.844 & 3.681\\
Projected & Initial & CUDD sift & 2.191 & 2.028 & 3.909\\
Projected & Initial & aGrUM cond. & 6.791 & 7.580 & 6.425\\
Projected & Initial & aGrUM joint & 5.289 & 4.657 & 5.514\\
Projected & Initial & Storm sparse & 2.541 & 2.151 & 2.661\\
Projected & Initial & Storm compact & 2.396 & 1.967 & 2.515\\
Projected & Repeat & CUDD fixed & 1.086 & 1.084 & 1.093\\
Projected & Repeat & CUDD sift & 1.087 & 1.098 & 1.107\\
Projected & Repeat & aGrUM cond. & 10.626 & 10.739 & 7.441\\
Projected & Repeat & aGrUM joint & 3.530 & 3.622 & 3.024\\
Projected & Repeat & Storm sparse & 2.478 & 1.780 & 2.100\\
Projected & Repeat & Storm compact & 1.457 & 1.092 & 1.344\\
Projected & Counts & CUDD fixed & 1.062 & 1.074 & 1.053\\
Projected & Counts & CUDD sift & 1.076 & 1.096 & 1.058\\
Projected & Counts & aGrUM cond. & 10.694 & 10.971 & 7.422\\
Projected & Counts & aGrUM joint & 3.537 & 3.661 & 3.045\\
Projected & Counts & Storm sparse & 2.673 & 2.031 & 2.202\\
Projected & Counts & Storm compact & 1.662 & 1.325 & 1.470\\
Projected & Masks & CUDD fixed & 1.060 & 1.093 & 1.036\\
Projected & Masks & CUDD sift & 1.032 & 1.086 & 1.033\\
Projected & Masks & aGrUM cond. & 12.910 & 13.052 & 9.551\\
Projected & Masks & aGrUM joint & 8.072 & 6.929 & 7.125\\
Projected & Masks & Storm sparse & 6.664 & 4.931 & 4.948\\
Projected & Masks & Storm compact & 6.025 & 4.589 & 4.633\\
Projected & Metadata & CUDD fixed & 1.041 & 1.058 & 1.030\\
Projected & Metadata & CUDD sift & 1.032 & 1.078 & 1.019\\
Projected & Metadata & aGrUM cond. & 9.871 & 11.131 & 7.360\\
Projected & Metadata & aGrUM joint & 3.331 & 3.699 & 2.988\\
Projected & Metadata & Storm sparse & 3.078 & 2.081 & 2.233\\
Projected & Metadata & Storm compact & 1.747 & 1.230 & 1.605\\
Projected & Add edge & CUDD fixed & 0.534 & 0.510 & 0.645\\
Projected & Add edge & CUDD sift & 0.613 & 0.614 & 0.788\\
Projected & Add edge & aGrUM cond. & 6.944 & 6.995 & 6.859\\
Projected & Add edge & aGrUM joint & 5.252 & 4.687 & 5.894\\
Projected & Add edge & Storm sparse & 1.959 & 1.715 & 1.940\\
Projected & Add edge & Storm compact & 2.055 & 1.251 & 1.979\\
Projected & Restore & CUDD fixed & 0.501 & 0.511 & 0.648\\
Projected & Restore & CUDD sift & 0.609 & 0.610 & 0.789\\
Projected & Restore & aGrUM cond. & 6.736 & 7.161 & 6.944\\
Projected & Restore & aGrUM joint & 5.339 & 4.723 & 6.018\\
Projected & Restore & Storm sparse & 1.770 & 1.661 & 1.802\\
Projected & Restore & Storm compact & 1.542 & 1.543 & 1.632\\
\end{longtable}

\endgroup

\begin{table}[htbp]\centering\small
\caption{Median process/children peak RSS in MiB, including import and all seven snapshots. This is not memory attributed to one query.}\label{tab:s-rss}
\begin{tabular}{@{}llrrr@{}}
\toprule
Input & Method & DeathStar & OTel & Petclinic\\
\midrule
Original & Direct & 36.1 & 34.4 & 36.1\\
Original & Prepared & 36.1 & 34.4 & 36.1\\
Original & CUDD fixed & 61.4 & 59.6 & 61.4\\
Original & CUDD sift & 61.3 & 59.6 & 61.4\\
Original & aGrUM cond. & 47.9 & 46.4 & 47.7\\
Original & aGrUM joint & 47.9 & 46.4 & 47.7\\
Original & Storm sparse & 94.7 & 93.3 & 94.0\\
Original & Storm compact & 94.8 & 93.3 & 94.1\\
Projected & Direct & 36.2 & 34.5 & 36.1\\
Projected & Prepared & 36.1 & 34.4 & 36.1\\
Projected & CUDD fixed & 61.4 & 59.6 & 61.4\\
Projected & CUDD sift & 61.4 & 59.6 & 61.3\\
Projected & aGrUM cond. & 47.8 & 46.2 & 47.7\\
Projected & aGrUM joint & 47.9 & 46.4 & 47.7\\
Projected & Storm sparse & 94.6 & 93.3 & 94.1\\
Projected & Storm compact & 94.6 & 93.2 & 94.0\\
\bottomrule
\end{tabular}

\end{table}

\section{Complete reconstruction benchmark}\label{app:pipeline}
Each source input is confirmation split/N/repetition zero. Copies produce 3,600, 7,200 and 14,400 attempts with unique request, trace and span identities, preserving durations, timestamps, joint observations and source relations. Record count varies at a fixed empirical law and graph, without increasing independent observation duration or service count.

Three technical repetitions per application/volume give 27 sequential pairs. The pipeline runners provide four logical CPUs, approximately 15.61 GiB RAM and Python 3.13.15. DeathStar and Petclinic use AMD EPYC 9V74 processors; OTel uses Intel Xeon 6973P-C. PMX extraction uses Java 11 and Palladio uses Java 17. Within each application, all volumes and method orders run on the same runner. The timing boundary starts from sealed role-prepared calibration telemetry with tools installed and ends at the first saved forecasts. The graph path constructs the graph, all six same-law variants, the source-adapted estimator and the frequency baseline. PMX performs its own native projection, independent extraction for conditional-local and inclusive variants, model collection, one application-only solver pass and first forecast serialization. Process startup and required I/O are included. Acquisition, installation, common role preparation, file transfer and all additional verification are excluded for both methods.

All nine application pairs precede extra graph replay, repeated PMX passes, known-probability controls and forecast verification; 57 verification-process records are retained separately. Every pair reproduces its method-specific forecasts, bounds and statuses within $10^{-12}$. CPU sums disjoint timed commands; memory is the maximum process/children peak across sequential stages, not total runner memory.

Figure~\ref{fig:pipeline} in the main article shows first-forecast wall times over the three paired repetitions.
\begin{table}[htbp]\centering\footnotesize
\caption{Pipeline resources and arithmetic launch-wait sensitivity. The final column subtracts 60 or 80 prescribed wait seconds from PMX within each pair before division by graph time.}\label{tab:s-pipeline-resources}
\begin{tabular}{@{}lrlll@{}}
\toprule
Application & Attempts & CPU, s: graph / PMX & Peak MiB: graph / PMX & Wait-subtracted ratio\\
\midrule
DeathStar & 3,600 & 2.95 / 245.83 & 208.2 / 2389.8 & 44.58\\
DeathStar & 7,200 & 4.70 / 279.48 & 306.6 / 3608.2 & 31.57\\
DeathStar & 14,400 & 8.44 / 373.46 & 521.0 / 4379.8 & 24.07\\
OTel & 3,600 & 3.62 / 236.29 & 444.7 / 1852.5 & 36.82\\
OTel & 7,200 & 6.44 / 398.91 & 787.9 / 2591.6 & 34.47\\
OTel & 14,400 & 12.58 / 1004.32 & 1473.7 / 4355.9 & 43.80\\
Petclinic & 3,600 & 2.28 / 206.00 & 234.2 / 1604.1 & 55.09\\
Petclinic & 7,200 & 3.46 / 235.80 & 366.2 / 2293.7 & 38.92\\
Petclinic & 14,400 & 5.77 / 327.16 & 630.7 / 3798.8 & 32.66\\
\bottomrule
\end{tabular}

\end{table}

PMX waits 20 seconds per operation before sending its command; JVM startup can overlap. Subtracting this wait is an accounting sensitivity, not a measured persistent-launcher runtime. This first-forecast experiment supersedes earlier coefficients that included extra verification and different bindings, so their change cannot be attributed solely to eligibility semantics.

\section{Separate synthetic computation studies}\label{app:synthetic}
The model-only reduction grid varies graph size from four to 48 mandatory nodes and the number of unknown controls from zero to ten. Sixteen grid records have one warmup and seven measured solutions each. All 112 measured solutions match analytical expectations. With ten unknown controls and masked mandatory completions, the reduced evaluator retains 2,048 candidate evaluations while the complete compatible-state count grows from $2^{14}$ to $2^{58}$. These full-state counts are combinatorial counts, not measured brute-force runtimes.

\begin{figure}[htbp]\centering
\includegraphics[width=\linewidth]{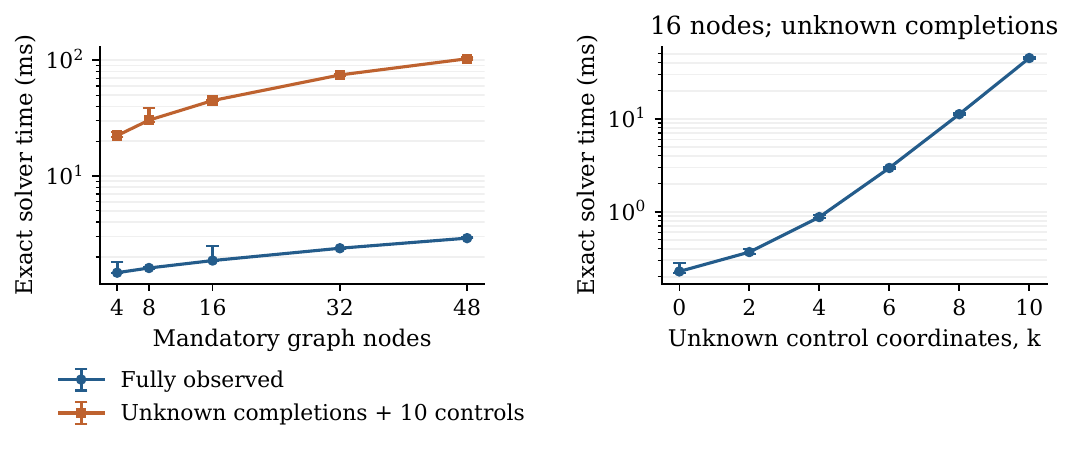}
\caption{Earlier model-only grid isolating graph evaluation and unknown-control growth. Inputs are already in memory; points summarize seven measured solutions.}\label{fig:s-grid}
\end{figure}

The earlier equivalent-backend synthetic grid has 27 positions varying graph size, total and unknown controls, category count and structure. It uses the same seven phases and common projection design. Its 9,072 planned synthetic measurements yield 8,964 successful answers, 24 aGrUM allocation failures on two encodings and 84 skipped later stream steps. All successful answers agree with the exact reference. These cases expose encoding limits and control-space growth; they are separate from the all-successful revised application panel.

Cases 00--04 are chains with 4, 8, 16, 32 and 48 nodes. Cases 05--09 vary completion coordinates over the same five counts with four graph nodes. Cases 10--15 have eight nodes and 0, 2, 4, 6, 8 or 10 unknown controls. Cases 16--19 have 5, 16, 64 or 128 observation categories. Cases 20--23 are 12-node chain, fanout, diamond and cyclic graphs. Cases 24--26 are fully observed 4-, 16- and 48-node graphs. All have ten total controls, hence 2,048 reduced table assignments. Initial incomplete cases have bounds $[1/5,3/5]$; fully observed cases have probability $1/16$. Later counts, masks and graph changes can change their answers.

\begin{figure}[htbp]\centering
\includegraphics[width=\linewidth]{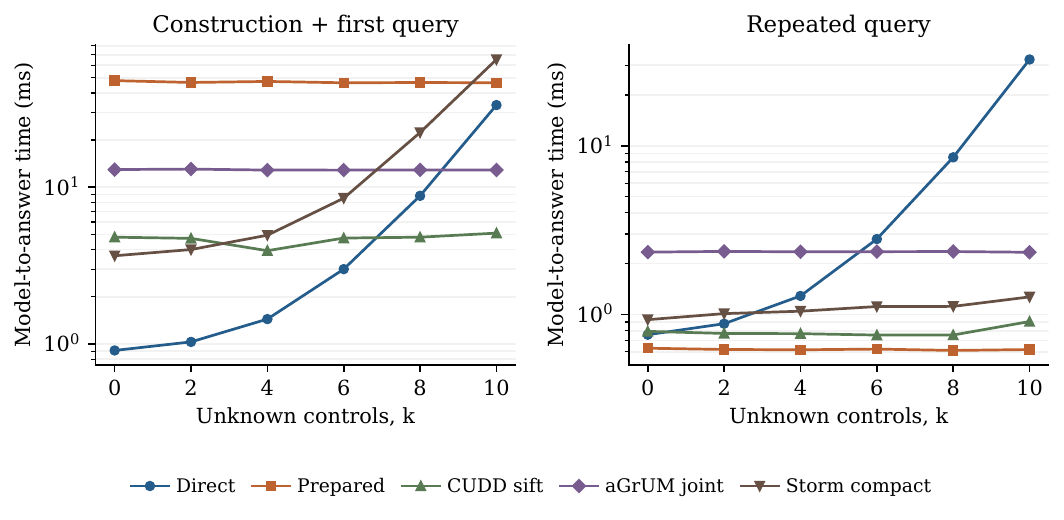}
\caption{Earlier synthetic comparison with eight graph nodes and ten total controls. Common completion-projected input is supplied to every solver, and projection cost is included. Each point summarizes three rounds.}\label{fig:s-synthetic}
\end{figure}

\clearpage
\begin{landscape}
\begingroup\footnotesize\setlength{\tabcolsep}{4pt}
\begin{longtable}{@{}lllrrrrrrrr@{}}
\caption{Every synthetic position and both inputs: initial, repeated and structural-update time in milliseconds (medians of three rounds). Dash: no qualified answer because the inference failed or the stream had already failed; it is not zero time. All seven phases are retained in the supplied CSV.}\label{tab:s-backend-synthetic}\\
\toprule
Case & Input & Phase & Direct & Prepared & CUDD fixed & CUDD sift & BN cond. & BN joint & Storm sparse & Storm compact\\
\midrule
\endfirsthead
\toprule
Case & Input & Phase & Direct & Prepared & CUDD fixed & CUDD sift & BN cond. & BN joint & Storm sparse & Storm compact\\
\midrule
\endhead
\midrule\multicolumn{11}{r}{Continued on next page}\\\endfoot
\bottomrule\endlastfoot
00 & Original & initial & 1.204 & 42.100 & 3.843 & 4.575 & 17.036 & 12.195 & 5.753 & 4.919\\
00 & Original & repeat & 0.939 & 0.448 & 0.522 & 0.577 & 8.845 & 2.163 & 2.143 & 0.935\\
00 & Original & structure & 0.970 & 44.599 & 0.601 & 1.000 & 17.137 & 11.963 & 4.143 & 3.205\\
00 & Projected & initial & 1.327 & 39.553 & 3.150 & 4.610 & 15.529 & 11.134 & 5.518 & 4.659\\
00 & Projected & repeat & 1.129 & 0.575 & 0.705 & 0.798 & 8.357 & 2.277 & 2.047 & 0.987\\
00 & Projected & structure & 1.054 & 42.363 & 0.690 & 1.016 & 15.906 & 11.130 & 3.831 & 2.974\\
01 & Original & initial & 1.322 & 54.208 & 3.848 & 4.705 & 20.516 & 15.505 & 6.520 & 5.718\\
01 & Original & repeat & 1.140 & 0.461 & 0.635 & 0.696 & 9.274 & 2.182 & 2.387 & 1.068\\
01 & Original & structure & 1.199 & 57.056 & 0.656 & 1.112 & 22.181 & 16.413 & 4.934 & 4.516\\
01 & Projected & initial & 1.443 & 47.002 & 3.308 & 4.008 & 17.301 & 12.860 & 5.784 & 4.922\\
01 & Projected & repeat & 1.268 & 0.630 & 0.686 & 0.769 & 8.441 & 2.354 & 2.108 & 1.057\\
01 & Projected & structure & 1.303 & 49.946 & 0.757 & 1.082 & 17.503 & 12.748 & 4.133 & 3.256\\
02 & Original & initial & 1.740 & 80.868 & 4.163 & 5.224 & 29.721 & 22.761 & 8.202 & 7.149\\
02 & Original & repeat & 1.569 & 0.511 & 0.628 & 0.698 & 11.436 & 2.194 & 2.927 & 1.420\\
02 & Original & structure & 1.699 & 84.022 & 0.838 & 1.425 & 31.628 & 24.858 & 7.027 & 5.350\\
02 & Projected & initial & 1.887 & 64.596 & 4.315 & 4.981 & 21.278 & 16.492 & 6.362 & 5.517\\
02 & Projected & repeat & 1.700 & 0.749 & 0.786 & 0.881 & 8.889 & 2.423 & 2.197 & 1.218\\
02 & Projected & structure & 1.784 & 67.070 & 0.938 & 1.287 & 22.584 & 17.542 & 4.675 & 3.716\\
03 & Original & initial & 2.548 & 131.208 & 3.749 & 4.935 & 47.612 & 40.147 & 11.647 & 10.972\\
03 & Original & repeat & 2.371 & 0.649 & 0.763 & 0.828 & 12.884 & 2.262 & 3.984 & 1.880\\
03 & Original & structure & 2.533 & 133.528 & 1.269 & 2.261 & 55.344 & 46.045 & 9.883 & 8.683\\
03 & Projected & initial & 2.753 & 95.163 & 4.784 & 5.440 & 30.214 & 24.985 & 8.087 & 7.205\\
03 & Projected & repeat & 2.623 & 0.998 & 0.995 & 1.088 & 9.720 & 2.712 & 2.438 & 1.356\\
03 & Projected & structure & 2.791 & 98.264 & 1.468 & 1.809 & 32.163 & 27.026 & 6.369 & 6.008\\
04 & Original & initial & 3.378 & 180.436 & 5.070 & 7.222 & 70.380 & 60.993 & 16.179 & 15.420\\
04 & Original & repeat & 3.194 & 0.755 & 0.914 & 0.981 & 14.686 & 2.315 & 5.046 & 2.362\\
04 & Original & structure & 4.214 & 180.991 & 1.931 & 3.201 & 83.639 & 73.170 & 13.843 & 13.238\\
04 & Projected & initial & 3.719 & 128.792 & 5.335 & 5.253 & 39.744 & 35.283 & 10.162 & 9.358\\
04 & Projected & repeat & 3.506 & 1.219 & 1.199 & 1.296 & 9.541 & 2.887 & 2.678 & 1.566\\
04 & Projected & structure & 3.708 & 130.355 & 2.075 & 2.416 & 43.989 & 38.898 & 9.057 & 7.607\\
05 & Original & initial & 1.224 & 44.180 & 3.877 & 4.457 & 17.146 & 12.184 & 5.755 & 4.863\\
05 & Original & repeat & 0.942 & 0.449 & 0.513 & 0.582 & 8.815 & 2.204 & 2.109 & 0.945\\
05 & Original & structure & 0.995 & 46.114 & 0.582 & 0.993 & 17.020 & 11.998 & 4.165 & 3.193\\
05 & Projected & initial & 1.305 & 38.429 & 3.975 & 4.602 & 15.589 & 11.172 & 5.445 & 4.670\\
05 & Projected & repeat & 1.071 & 0.567 & 0.716 & 0.790 & 8.389 & 2.317 & 2.111 & 1.037\\
05 & Projected & structure & 1.072 & 41.173 & 0.696 & 1.006 & 15.709 & 11.101 & 3.777 & 3.025\\
06 & Original & initial & 1.211 & 57.073 & 3.092 & 3.795 & 21.330 & 15.844 & 6.551 & 5.614\\
06 & Original & repeat & 1.058 & 0.451 & 0.561 & 0.616 & 9.854 & 2.162 & 2.410 & 1.032\\
06 & Original & structure & 1.085 & 59.153 & 0.647 & 1.138 & 22.679 & 16.633 & 4.839 & 4.537\\
06 & Projected & initial & 1.321 & 48.483 & 4.122 & 3.871 & 18.092 & 13.161 & 5.792 & 4.872\\
06 & Projected & repeat & 1.146 & 0.615 & 0.685 & 0.761 & 8.809 & 2.321 & 2.124 & 0.984\\
06 & Projected & structure & 1.193 & 50.863 & 0.744 & 1.073 & 18.224 & 13.189 & 4.116 & 3.200\\
07 & Original & initial & 1.398 & 74.222 & 4.068 & 4.151 & 27.533 & 21.157 & 7.780 & 6.743\\
07 & Original & repeat & 1.228 & 0.502 & 0.623 & 0.691 & 10.976 & 2.198 & 2.925 & 1.274\\
07 & Original & structure & 1.259 & 77.545 & 0.765 & 1.320 & 29.828 & 23.166 & 6.616 & 5.033\\
07 & Projected & initial & 1.538 & 56.605 & 3.402 & 3.972 & 19.741 & 15.040 & 6.114 & 5.146\\
07 & Projected & repeat & 1.354 & 0.695 & 0.745 & 0.843 & 8.896 & 2.437 & 2.229 & 1.083\\
07 & Projected & structure & 1.376 & 60.121 & 0.838 & 1.154 & 20.832 & 15.792 & 4.377 & 3.402\\
08 & Original & initial & 1.734 & 105.913 & 4.373 & 5.459 & 40.210 & 32.547 & 10.434 & 9.641\\
08 & Original & repeat & 1.562 & 0.657 & 0.753 & 0.822 & 12.941 & 2.219 & 4.009 & 1.839\\
08 & Original & structure & 1.701 & 109.474 & 0.922 & 1.835 & 44.960 & 36.947 & 8.468 & 7.369\\
08 & Projected & initial & 1.895 & 69.657 & 4.464 & 4.947 & 23.184 & 18.520 & 6.680 & 5.755\\
08 & Projected & repeat & 1.754 & 0.935 & 0.931 & 1.006 & 8.936 & 2.595 & 2.403 & 1.245\\
08 & Projected & structure & 1.794 & 72.560 & 1.082 & 1.421 & 25.130 & 19.927 & 4.958 & 4.631\\
09 & Original & initial & 2.160 & 140.373 & 4.541 & 6.092 & 54.767 & 45.836 & 13.298 & 12.497\\
09 & Original & repeat & 1.956 & 0.671 & 0.870 & 0.941 & 15.306 & 2.268 & 5.065 & 2.258\\
09 & Original & structure & 2.061 & 145.005 & 1.184 & 2.318 & 64.457 & 54.333 & 10.846 & 10.205\\
09 & Projected & initial & 2.291 & 83.304 & 4.560 & 5.192 & 27.519 & 22.502 & 7.265 & 6.395\\
09 & Projected & repeat & 2.136 & 1.078 & 1.120 & 1.165 & 9.562 & 2.807 & 2.577 & 1.400\\
09 & Projected & structure & 2.184 & 86.381 & 1.242 & 1.547 & 29.624 & 24.324 & 6.150 & 4.626\\
10 & Original & initial & 0.767 & 54.996 & 3.074 & 3.827 & 20.608 & 15.272 & 4.423 & 3.896\\
10 & Original & repeat & 0.620 & 0.456 & 0.629 & 0.708 & 9.401 & 2.138 & 1.624 & 0.856\\
10 & Original & structure & 0.650 & 57.341 & 0.667 & 1.129 & 22.024 & 16.365 & 2.717 & 2.182\\
10 & Projected & initial & 0.910 & 47.916 & 3.338 & 4.803 & 17.591 & 12.975 & 4.015 & 3.656\\
10 & Projected & repeat & 0.758 & 0.630 & 0.687 & 0.793 & 8.576 & 2.340 & 1.588 & 0.931\\
10 & Projected & structure & 0.690 & 50.496 & 0.698 & 1.003 & 15.625 & 12.219 & 2.125 & 1.854\\
11 & Original & initial & 0.864 & 56.211 & 3.891 & 4.599 & 20.543 & 15.438 & 5.268 & 4.576\\
11 & Original & repeat & 0.717 & 0.464 & 0.643 & 0.704 & 9.334 & 2.164 & 1.998 & 1.035\\
11 & Original & structure & 0.764 & 58.703 & 0.670 & 1.115 & 22.065 & 16.264 & 3.555 & 2.887\\
11 & Projected & initial & 1.031 & 46.626 & 3.311 & 4.721 & 17.412 & 13.059 & 4.469 & 4.005\\
11 & Projected & repeat & 0.881 & 0.619 & 0.685 & 0.771 & 8.556 & 2.363 & 1.641 & 1.011\\
11 & Projected & structure & 0.904 & 49.259 & 0.800 & 1.083 & 17.507 & 12.948 & 2.717 & 2.274\\
12 & Original & initial & 1.342 & 54.730 & 3.907 & 4.624 & 20.655 & 15.359 & 6.511 & 5.579\\
12 & Original & repeat & 1.176 & 0.468 & 0.636 & 0.696 & 9.297 & 2.170 & 2.383 & 1.062\\
12 & Original & structure & 1.220 & 57.316 & 0.666 & 1.106 & 22.224 & 16.190 & 4.813 & 4.452\\
12 & Projected & initial & 1.442 & 47.336 & 4.226 & 3.938 & 17.392 & 12.872 & 5.822 & 4.944\\
12 & Projected & repeat & 1.287 & 0.617 & 0.682 & 0.769 & 8.479 & 2.350 & 2.159 & 1.045\\
12 & Projected & structure & 1.381 & 49.457 & 0.766 & 1.075 & 17.504 & 12.779 & 4.126 & 3.233\\
13 & Original & initial & 2.943 & 55.326 & 3.875 & 4.583 & 20.691 & 15.281 & 11.710 & 9.454\\
13 & Original & repeat & 2.781 & 0.457 & 0.641 & 0.685 & 9.264 & 2.174 & 3.952 & 1.103\\
13 & Original & structure & 2.855 & 57.645 & 0.688 & 1.104 & 22.148 & 16.242 & 10.111 & 7.679\\
13 & Projected & initial & 3.004 & 46.418 & 4.070 & 4.742 & 17.430 & 12.867 & 10.521 & 8.517\\
13 & Projected & repeat & 2.798 & 0.623 & 0.680 & 0.754 & 8.510 & 2.355 & 3.673 & 1.113\\
13 & Projected & structure & 2.917 & 49.112 & 0.781 & 1.088 & 17.439 & 12.816 & 9.622 & 6.612\\
14 & Original & initial & 9.226 & 54.660 & 3.954 & 4.709 & 20.602 & 15.328 & 31.173 & 24.566\\
14 & Original & repeat & 9.016 & 0.467 & 0.620 & 0.687 & 9.288 & 2.161 & 9.263 & 1.125\\
14 & Original & structure & 9.868 & 57.064 & 0.666 & 1.102 & 22.057 & 16.387 & 29.462 & 22.447\\
14 & Projected & initial & 8.821 & 46.600 & 3.293 & 4.803 & 17.390 & 12.899 & 28.527 & 22.355\\
14 & Projected & repeat & 8.536 & 0.611 & 0.672 & 0.755 & 8.469 & 2.359 & 8.995 & 1.116\\
14 & Projected & structure & 8.704 & 49.098 & 0.755 & 1.069 & 17.441 & 12.773 & 26.900 & 20.448\\
15 & Original & initial & 35.664 & 55.160 & 4.234 & 4.864 & 20.566 & 15.345 & 92.100 & 74.789\\
15 & Original & repeat & 34.673 & 0.454 & 0.714 & 0.749 & 9.394 & 2.172 & 20.703 & 1.337\\
15 & Original & structure & 35.585 & 57.383 & 0.794 & 1.242 & 22.165 & 16.374 & 99.527 & 82.090\\
15 & Projected & initial & 33.452 & 46.433 & 4.415 & 5.098 & 17.327 & 12.896 & 80.889 & 65.384\\
15 & Projected & repeat & 32.476 & 0.619 & 0.809 & 0.907 & 8.452 & 2.335 & 20.260 & 1.270\\
15 & Projected & structure & 33.518 & 48.649 & 0.886 & 1.217 & 17.522 & 12.688 & 88.073 & 71.719\\
16 & Original & initial & 1.395 & 75.106 & 4.094 & 4.998 & 27.695 & 20.892 & 7.842 & 6.704\\
16 & Original & repeat & 1.213 & 0.503 & 0.601 & 0.699 & 11.066 & 2.189 & 2.938 & 1.275\\
16 & Original & structure & 1.259 & 77.176 & 0.726 & 1.306 & 29.892 & 22.735 & 6.635 & 5.058\\
16 & Projected & initial & 1.520 & 55.427 & 4.187 & 4.908 & 19.602 & 14.921 & 6.103 & 5.136\\
16 & Projected & repeat & 1.347 & 0.689 & 0.757 & 0.827 & 8.854 & 2.433 & 2.206 & 1.084\\
16 & Projected & structure & 1.384 & 58.158 & 0.834 & 1.194 & 21.009 & 15.629 & 4.393 & 3.383\\
17 & Original & initial & 8.865 & 74.613 & 5.158 & 5.227 & 58.391 & 32.251 & 42.818 & 31.206\\
17 & Original & repeat & 8.663 & 1.188 & 1.496 & 1.646 & 34.196 & 6.336 & 16.492 & 1.498\\
17 & Original & structure & 8.832 & 77.211 & 1.676 & 2.288 & 60.863 & 34.151 & 40.303 & 29.791\\
17 & Projected & initial & 1.687 & 55.299 & 4.375 & 4.942 & 19.685 & 15.173 & 6.309 & 5.333\\
17 & Projected & repeat & 1.514 & 0.842 & 0.897 & 0.977 & 8.904 & 2.572 & 2.370 & 1.274\\
17 & Projected & structure & 1.610 & 58.145 & 0.976 & 1.320 & 20.891 & 15.904 & 4.463 & 3.527\\
18 & Original & initial & 40.544 & 76.902 & 9.719 & 9.633 & 190.871 & 79.841 & 198.792 & 140.643\\
18 & Original & repeat & 40.025 & 5.004 & 5.575 & 6.011 & 136.206 & 24.568 & 72.604 & 1.617\\
18 & Original & structure & 40.912 & 79.885 & 5.966 & 6.826 & 196.979 & 82.428 & 186.415 & 132.243\\
18 & Projected & initial & 2.276 & 58.574 & 4.867 & 5.493 & 20.292 & 15.401 & 6.916 & 6.000\\
18 & Projected & repeat & 2.130 & 1.428 & 1.489 & 1.571 & 9.569 & 3.099 & 2.937 & 1.887\\
18 & Projected & structure & 2.912 & 61.982 & 1.588 & 1.914 & 21.654 & 16.456 & 5.078 & 4.307\\
19 & Original & initial & 84.052 & 82.708 & 15.764 & 15.891 & 374.908 & 144.101 & 397.087 & 276.993\\
19 & Original & repeat & 82.188 & 8.101 & 11.231 & 11.844 & 283.668 & 49.343 & 147.702 & 1.742\\
19 & Original & structure & 83.725 & 84.986 & 11.769 & 13.104 & 383.208 & 147.350 & 394.315 & 274.581\\
19 & Projected & initial & 3.150 & 56.963 & 4.961 & 5.477 & 21.189 & 16.337 & 7.508 & 6.682\\
19 & Projected & repeat & 2.904 & 2.199 & 2.262 & 2.349 & 10.709 & 3.924 & 3.720 & 2.701\\
19 & Projected & structure & 3.090 & 59.905 & 2.464 & 2.822 & 22.444 & 17.414 & 6.639 & 4.979\\
20 & Original & initial & 1.547 & 66.744 & 4.123 & 5.041 & 24.860 & 18.874 & 7.407 & 6.332\\
20 & Original & repeat & 1.382 & 0.502 & 0.593 & 0.656 & 9.931 & 2.174 & 2.651 & 1.176\\
20 & Original & structure & 1.441 & 70.027 & 0.723 & 1.335 & 26.828 & 20.420 & 5.688 & 4.586\\
20 & Projected & initial & 1.685 & 56.554 & 4.219 & 4.051 & 19.310 & 14.605 & 6.135 & 5.186\\
20 & Projected & repeat & 1.505 & 0.681 & 0.742 & 0.828 & 8.892 & 2.439 & 2.171 & 1.103\\
20 & Projected & structure & 1.567 & 57.637 & 0.851 & 1.186 & 20.057 & 15.375 & 4.473 & 3.396\\
21 & Original & initial & 1.736 & 148.130 & 4.185 & 4.969 & 52.827 & 46.315 & 10.324 & 9.218\\
21 & Original & repeat & 1.569 & 0.595 & 0.652 & 0.729 & 11.177 & 2.214 & 3.267 & 1.496\\
21 & Original & structure & 1.676 & 161.121 & 1.349 & 2.065 & 64.168 & 57.350 & 9.368 & 7.642\\
21 & Projected & initial & 1.883 & 125.542 & 4.335 & 4.975 & 40.119 & 35.367 & 7.901 & 7.324\\
21 & Projected & repeat & 1.750 & 0.794 & 0.827 & 0.917 & 8.974 & 2.480 & 2.279 & 1.171\\
21 & Projected & structure & 1.810 & 140.323 & 0.923 & 1.229 & 48.222 & 43.017 & 6.382 & 5.393\\
22 & Original & initial & 1.747 & 347.993 & 4.316 & 5.072 & --- & --- & 14.200 & 13.779\\
22 & Original & repeat & 1.581 & 0.578 & 0.659 & 0.756 & --- & --- & 3.412 & 1.522\\
22 & Original & structure & 1.781 & 350.111 & 1.395 & 2.159 & --- & --- & 12.607 & 11.278\\
22 & Projected & initial & 1.913 & 324.133 & 4.323 & 4.971 & --- & --- & 11.932 & 11.075\\
22 & Projected & repeat & 1.751 & 1.787 & 0.846 & 0.925 & --- & --- & 2.348 & 1.202\\
22 & Projected & structure & 1.856 & 329.310 & 0.973 & 1.280 & --- & --- & 10.101 & 9.202\\
23 & Original & initial & 1.757 & 196.817 & 3.421 & 5.271 & 72.540 & 61.651 & 11.097 & 10.524\\
23 & Original & repeat & 1.592 & 0.557 & 0.671 & 0.751 & 15.776 & 2.207 & 3.398 & 1.517\\
23 & Original & structure & 1.762 & 354.477 & 0.972 & 1.610 & --- & --- & 12.383 & 11.450\\
23 & Projected & initial & 1.929 & 163.003 & 4.335 & 5.028 & 57.920 & 49.522 & 8.663 & 7.706\\
23 & Projected & repeat & 1.740 & 0.795 & 0.839 & 0.928 & 13.119 & 2.515 & 2.330 & 1.230\\
23 & Projected & structure & 1.838 & 310.422 & 1.006 & 1.404 & --- & --- & 10.719 & 8.916\\
24 & Original & initial & 1.568 & 32.037 & 4.748 & 4.461 & 27.719 & 15.943 & 3.111 & 3.054\\
24 & Original & repeat & 1.401 & 1.086 & 1.365 & 1.580 & 17.832 & 4.926 & 0.533 & 0.210\\
24 & Original & structure & 0.842 & 30.748 & 0.758 & 1.077 & 15.119 & 10.088 & 1.540 & 1.411\\
24 & Projected & initial & 0.596 & 27.511 & 2.934 & 3.137 & 6.624 & 6.093 & 1.733 & 1.757\\
24 & Projected & repeat & 0.491 & 0.465 & 0.534 & 0.553 & 2.442 & 1.014 & 0.465 & 0.412\\
24 & Projected & structure & 0.388 & 26.998 & 0.424 & 0.626 & 5.929 & 5.429 & 0.956 & 0.870\\
25 & Original & initial & 1.964 & 46.705 & 5.224 & 5.824 & 38.855 & 23.198 & 4.018 & 3.893\\
25 & Original & repeat & 1.778 & 1.200 & 1.626 & 1.613 & 21.932 & 4.889 & 0.661 & 0.258\\
25 & Original & structure & 1.283 & 45.977 & 1.021 & 1.594 & 22.800 & 15.658 & 2.257 & 2.475\\
25 & Projected & initial & 0.855 & 28.272 & 4.133 & 4.375 & 6.868 & 6.349 & 2.123 & 2.168\\
25 & Projected & repeat & 0.715 & 0.661 & 0.660 & 0.652 & 2.608 & 1.193 & 0.705 & 0.618\\
25 & Projected & structure & 0.676 & 27.217 & 0.688 & 0.854 & 6.297 & 5.719 & 1.218 & 1.201\\
26 & Original & initial & 2.964 & 86.017 & 5.704 & 7.582 & 74.964 & 46.826 & 6.701 & 6.675\\
26 & Original & repeat & 2.715 & 1.436 & 2.849 & 2.804 & 35.731 & 5.019 & 0.920 & 0.384\\
26 & Original & structure & 2.100 & 84.139 & 2.155 & 3.298 & 48.888 & 36.055 & 4.384 & 4.032\\
26 & Projected & initial & 1.516 & 30.265 & 5.306 & 4.664 & 8.800 & 8.211 & 3.955 & 3.960\\
26 & Projected & repeat & 1.328 & 1.211 & 1.187 & 1.168 & 3.131 & 1.721 & 1.198 & 1.137\\
26 & Projected & structure & 1.199 & 29.179 & 1.770 & 1.961 & 7.998 & 7.521 & 2.989 & 2.966\\
\end{longtable}

\endgroup
\end{landscape}

\section{Direct identification of the completion-and-deadline target}\label{app:direct}
\subsection{Inputs, support and target preservation}
The eighteen retained confirmation campaigns supply sixty operation/campaign positions each for calibration and test. Four constructions evaluate full $E$, $R$ on a full model, full-model projection to $(T,U)$ followed by $R=TU$, and direct $(T,U)$ construction from native records. Full construction reproduces the frozen model; direct construction shares native-parent, logical-retry and time interpretation without building a full replica state. Interpretation excludes external outcome fields.

The three $R$ paths agree exactly on primitive observations, projected laws, endpoints and statuses on common support, including after masking. Calibration yields 54 points, five intervals and one refusal; test reconstruction supports 57 positions. One OTel checkout calibration case and three DeathStar user-timeline test cases retain unexplained native-parent boundaries. Test-state refusal does not invalidate a calibration forecast or the independently verified test frequency. Forecast accuracy therefore retains the confirmation support and aggregation.

\begin{table}[htbp]\centering\small
\caption{Retrospective forecast comparison on the unchanged calibration point support. The full $E$ and direct $R$ probabilities, five interval endpoint pairs, signed errors and Brier scores agree exactly.}\label{tab:s-direct-accuracy}
\begin{tabular}{@{}lrrrrr@{}}
\toprule
Application & Points & Intervals & Refusals & MAE $E$, pp & MAE $R$, pp\\
\midrule
DeathStar & 18 & 0 & 0 & 0.944444 & 0.944444\\
OTel & 12 & 5 & 1 & 2.000000 & 2.000000\\
Petclinic & 24 & 0 & 0 & 0.925926 & 0.925926\\
\bottomrule
\end{tabular}

\end{table}

There are 116 supported calibration/test model pairs. Every pair has matching separate $E$ and $R$ intervals on unmasked records. This is different from equality on their joint compatible-state family. Among the 61,200 supported test attempts, 56,791 force $R-E=0$ and 4,409 permit a difference: 1,382 in DeathStar and 3,027 in OTel. No attempt forces a positive difference. For DeathStar compose-post, NC, repetition zero, 790 of 1,200 attempts force equality and 410 leave it unresolved. Both events have interval $[11/20,107/120]$, while their difference has interval $[0,41/120]$. Two OTel attempts permit either success outcome but force equality of the events. These account for the difference between 4,411 ambiguous successes and 4,409 ambiguous structural contrasts.

The protocol preceded these retrospective summaries but followed availability of the test outcomes. It checks the new construction against frozen forecasts and verified outcomes, preserving the original bindings and labels; it adds no independent prediction sample.

\subsection{Primitive masking and the information carried by structure}
Primitive coordinates are masked before conjunction. The 57 supported test models, six levels, three mechanisms and five reveal settings yield 5,130 rows. Failure-associated masking prioritizes observed failed logical completions, without the external $Y$ used by the synthetic mechanism in Section~\ref{app:missingness}. Uniform and attempt-grouped mechanisms complete the design; the final group can be split to match the coordinate budget. Original unknowns, attempt census, deadlines and graph/group declarations stay fixed.

At 50\% loss, every mechanism hides 220,803 of 441,619 known primitive coordinates, with integer rounding within models. The full and reduced estimators receive the same hidden cells. Uniform masking produces 1,437 attempts with established $R=1$ and unknown structural conditions; the failure-associated mask produces 1,598. These cases illustrate the two-coordinate image result and the loss of $E$ information when only $R$ is retained. All models remain compatible with the external outcomes in both directions after coarsening, as required by fiber inclusion.

Both $E/R$ endpoints coincide in 4,227 of 5,130 case/settings rows; in 903, only the $E$ lower endpoint is smaller. These respective-event ranges describe repeated coarsenings of the same supported cases.

\begin{table}[htbp]\centering\small
\caption{Information groups restored after 50\% uniform masking. Widths are attempt-weighted percentage points; the upper count sums attempts permitting $D=R-E=1$. Restoration groups have different actual coordinate budgets.}\label{tab:s-direct-reveal}
\begin{tabular}{@{}lrrrr@{}}
\toprule
Restored group & Cells & Width $E$, pp & Width $R$, pp & Upper count $D$\\
\midrule
None & 0 & 87.420 & 85.072 & 34,109\\
Selection and receiver & 56,521 & 85.072 & 85.072 & 4,689\\
Entry completion & 30,636 & 82.098 & 77.438 & 33,950\\
Mandatory completions & 133,646 & 73.843 & 49.944 & 33,956\\
All hidden evidence & 220,803 & 7.208 & 7.208 & 4,409\\
\bottomrule
\end{tabular}

\end{table}

Selection/receiver restoration improves the structural contrast while leaving $R$ unchanged; completion restoration narrows $R$ itself. Unequal budgets preclude an equal-cost strategy ranking. Full restoration recovers the original probabilities and 4,409 unresolved differences.

\subsection{First answers, repeated queries and stage costs}
All four routes share one runner per application, warm file caches, fresh processes and rotated order. Copies of split/N/repetition-zero inputs yield 3,600, 7,200 and 14,400 attempts with unique identifiers and unchanged empirical law; three rounds give 108 processes. Full $E$ uses exact completion reduction; full $R$ queries only relevant coordinates. Projection and direct $R$ share their final representation and query code. Independent verification follows all application measurements.

The wall-time boundary starts before process launch and ends with the first saved model and answer. Stage records separate reading, binding, representation, count aggregation, explicit reduction, query, serialization and residual startup/dispatch. CPU covers the same process lifecycle, and peak RSS includes decoded native inputs and imported packages. All 108 processes qualify, and every one of the 27 paired positions favours direct construction. The median of three paired time reductions per workload gives the descriptive 5.08--11.16\% range against full $E$. At 3,600 attempts, the median paired absolute reductions are 0.114, 0.263 and 0.121 seconds for DeathStar, OTel and Petclinic, respectively. The corresponding ranges for the same $R$ target are 4.70--11.31\% against full construction and 4.69--11.71\% against full construction plus projection.

Across operations, full and direct coordinate totals are respectively 29/6 for DeathStar, 38/6 for OTel and 25/8 for Petclinic. Observed-category totals are 10/6, 22/7 and 13/8. At the original volume, all routes read 53,155, 106,470 and 34,964 native spans for the three applications. Keeping these native inputs in memory explains why peak RSS changes little despite a smaller identified representation.

\begin{table}[htbp]\centering\footnotesize
\caption{Ready-model query times in microseconds. Each cell is the range of operation/volume medians across three rounds, with 31 repeated queries per ready model. Total retained queries: 11,160.}\label{tab:s-direct-ready}
\begin{tabular}{@{}lrrrr@{}}
\toprule
Application & Full $E$ & Full $R$ & Projected $R$ & Direct $R$\\
\midrule
DeathStar & 22.287--1166.467 & 5.086--16.540 & 3.986--5.599 & 3.915--5.655\\
OTel & 375.080--2254.557 & 12.373--34.795 & 5.060--5.831 & 5.119--5.931\\
Petclinic & 20.388--807.384 & 5.169--13.135 & 4.087--5.911 & 4.077--5.951\\
\bottomrule
\end{tabular}

\end{table}

Stage values are workload medians of three fresh processes; 0.000 seconds denotes sub-millisecond work, retained unrounded in CSV. The PMX ratios concern a separate first-forecast experiment with different outputs and are not multiplied by these four-route ratios.

\clearpage
\begin{landscape}
\begingroup\footnotesize\setlength{\tabcolsep}{3.5pt}
\begin{longtable}{@{}llrrrrrrrrrrr@{}}
\caption{Full stage and resource medians for the four matched identification routes. Stage medians need not add to the median total. Times and CPU are seconds; RSS is MiB.}\\
\toprule
App./$n$ & Route & Read & Bind & State & Count & Reduce & Query & Save & Start & Total & CPU & RSS\\
\midrule\endfirsthead
\toprule
App./$n$ & Route & Read & Bind & State & Count & Reduce & Query & Save & Start & Total & CPU & RSS\\
\midrule\endhead
DeathStar/3,600 & $E$: full & 0.228 & 0.774 & 0.173 & 0.010 & 0.000 & 0.002 & 0.001 & 0.772 & 1.956 & 2.530 & 206.4\\
DeathStar/3,600 & $R$: full & 0.228 & 0.772 & 0.175 & 0.010 & 0.000 & 0.000 & 0.001 & 0.770 & 1.959 & 2.531 & 206.3\\
DeathStar/3,600 & $R$: projected & 0.227 & 0.778 & 0.173 & 0.010 & 0.000 & 0.000 & 0.000 & 0.768 & 1.958 & 2.525 & 206.4\\
DeathStar/3,600 & $R$: direct & 0.226 & 0.748 & 0.094 & 0.008 & 0.000 & 0.000 & 0.000 & 0.765 & 1.840 & 2.411 & 206.6\\
DeathStar/7,200 & $E$: full & 0.464 & 1.547 & 0.504 & 0.021 & 0.000 & 0.002 & 0.001 & 0.771 & 3.317 & 3.888 & 310.7\\
DeathStar/7,200 & $R$: full & 0.462 & 1.539 & 0.511 & 0.022 & 0.000 & 0.000 & 0.001 & 0.772 & 3.318 & 3.891 & 310.5\\
DeathStar/7,200 & $R$: projected & 0.468 & 1.543 & 0.511 & 0.022 & 0.000 & 0.000 & 0.000 & 0.773 & 3.318 & 3.886 & 310.7\\
DeathStar/7,200 & $R$: direct & 0.466 & 1.571 & 0.280 & 0.015 & 0.000 & 0.000 & 0.000 & 0.769 & 3.103 & 3.674 & 310.5\\
DeathStar/14,400 & $E$: full & 0.940 & 3.639 & 0.782 & 0.040 & 0.000 & 0.002 & 0.001 & 0.774 & 6.186 & 6.759 & 519.2\\
DeathStar/14,400 & $R$: full & 0.943 & 3.654 & 0.780 & 0.040 & 0.000 & 0.000 & 0.001 & 0.778 & 6.200 & 6.774 & 519.2\\
DeathStar/14,400 & $R$: projected & 0.942 & 3.648 & 0.780 & 0.039 & 0.000 & 0.000 & 0.000 & 0.778 & 6.196 & 6.769 & 519.1\\
DeathStar/14,400 & $R$: direct & 0.931 & 3.174 & 0.591 & 0.029 & 0.000 & 0.000 & 0.000 & 0.772 & 5.495 & 6.068 & 519.2\\
OTel/3,600 & $E$: full & 0.764 & 2.353 & 0.352 & 0.018 & 0.000 & 0.003 & 0.001 & 0.863 & 4.377 & 5.011 & 446.4\\
OTel/3,600 & $R$: full & 0.759 & 2.324 & 0.353 & 0.019 & 0.000 & 0.000 & 0.001 & 0.870 & 4.311 & 4.917 & 446.4\\
OTel/3,600 & $R$: projected & 0.760 & 2.308 & 0.350 & 0.018 & 0.000 & 0.000 & 0.000 & 0.859 & 4.322 & 4.955 & 446.5\\
OTel/3,600 & $R$: direct & 0.758 & 2.323 & 0.158 & 0.009 & 0.000 & 0.000 & 0.000 & 0.861 & 4.113 & 4.748 & 446.6\\
OTel/7,200 & $E$: full & 1.534 & 4.647 & 0.757 & 0.027 & 0.000 & 0.003 & 0.001 & 0.872 & 7.832 & 8.465 & 789.2\\
OTel/7,200 & $R$: full & 1.550 & 4.687 & 0.759 & 0.026 & 0.000 & 0.000 & 0.001 & 0.873 & 7.893 & 8.527 & 789.5\\
OTel/7,200 & $R$: projected & 1.561 & 4.656 & 0.762 & 0.026 & 0.001 & 0.000 & 0.000 & 0.875 & 7.872 & 8.516 & 789.4\\
OTel/7,200 & $R$: direct & 1.520 & 4.426 & 0.586 & 0.017 & 0.000 & 0.000 & 0.000 & 0.868 & 7.434 & 8.062 & 789.4\\
OTel/14,400 & $E$: full & 3.054 & 10.052 & 1.525 & 0.053 & 0.000 & 0.004 & 0.001 & 0.892 & 15.592 & 16.167 & 1475.0\\
OTel/14,400 & $R$: full & 3.062 & 9.954 & 1.513 & 0.053 & 0.000 & 0.000 & 0.001 & 0.891 & 15.464 & 16.058 & 1475.0\\
OTel/14,400 & $R$: projected & 3.064 & 10.032 & 1.515 & 0.053 & 0.001 & 0.000 & 0.000 & 0.881 & 15.574 & 16.206 & 1474.9\\
OTel/14,400 & $R$: direct & 3.093 & 9.565 & 1.177 & 0.032 & 0.000 & 0.000 & 0.000 & 0.870 & 14.738 & 15.372 & 1474.9\\
Petclinic/3,600 & $E$: full & 0.284 & 0.749 & 0.135 & 0.009 & 0.000 & 0.002 & 0.001 & 0.853 & 2.035 & 2.654 & 238.5\\
Petclinic/3,600 & $R$: full & 0.284 & 0.752 & 0.138 & 0.009 & 0.000 & 0.000 & 0.001 & 0.857 & 2.051 & 2.678 & 238.7\\
Petclinic/3,600 & $R$: projected & 0.285 & 0.745 & 0.138 & 0.009 & 0.000 & 0.000 & 0.000 & 0.857 & 2.028 & 2.659 & 238.6\\
Petclinic/3,600 & $R$: direct & 0.281 & 0.702 & 0.066 & 0.006 & 0.000 & 0.000 & 0.000 & 0.853 & 1.908 & 2.541 & 238.4\\
Petclinic/7,200 & $E$: full & 0.576 & 1.604 & 0.307 & 0.018 & 0.000 & 0.002 & 0.001 & 0.861 & 3.366 & 3.993 & 373.4\\
Petclinic/7,200 & $R$: full & 0.578 & 1.614 & 0.299 & 0.018 & 0.000 & 0.000 & 0.001 & 0.858 & 3.368 & 3.989 & 373.2\\
Petclinic/7,200 & $R$: projected & 0.576 & 1.604 & 0.301 & 0.018 & 0.000 & 0.000 & 0.000 & 0.853 & 3.357 & 3.984 & 373.5\\
Petclinic/7,200 & $R$: direct & 0.580 & 1.368 & 0.333 & 0.013 & 0.000 & 0.000 & 0.000 & 0.852 & 3.143 & 3.770 & 373.3\\
Petclinic/14,400 & $E$: full & 1.160 & 3.461 & 0.646 & 0.036 & 0.000 & 0.002 & 0.001 & 0.867 & 6.179 & 6.807 & 643.0\\
Petclinic/14,400 & $R$: full & 1.160 & 3.478 & 0.644 & 0.036 & 0.000 & 0.000 & 0.001 & 0.866 & 6.193 & 6.790 & 643.2\\
Petclinic/14,400 & $R$: projected & 1.165 & 3.483 & 0.644 & 0.035 & 0.000 & 0.000 & 0.000 & 0.868 & 6.216 & 6.845 & 643.1\\
Petclinic/14,400 & $R$: direct & 1.165 & 2.963 & 0.467 & 0.025 & 0.000 & 0.000 & 0.000 & 0.857 & 5.486 & 6.090 & 643.2\\
\bottomrule
\end{longtable}

\endgroup
\end{landscape}

\subsection{Source and integrity records}
The retained sources are the completion-target run \texttt{34750418515} and the matched four-route follow-up \texttt{34751473713}. The complete archival retention record \texttt{34751987341} lists eighteen ZIP bundles totalling 17,256,100 bytes, including failed-attempt diagnostics and successful computation bundles. Only qualified measurements supply the reported results. The numerical package's \texttt{data/v5/source-manifest.json} verifies twenty local numerical files. The original 25-entry index, including five protocol/report records, is preserved separately under \texttt{provenance/}; \texttt{docs/SOURCES.md} links their original identities and current English documents. The native inventory identifies the complete evidence and serialized models.

\section{Analytical detail and finite verification}
\subsection{Error decomposition and serial dependence}
Use the main article's source observation law $Q_s$, source and target complete-state laws $P_s,P_t$, and fiber extrema $\ell(o),u(o)\in[0,1]$. On sound bindings, $\ell(O)\le E(W)\le u(O)$. The exceptional mass $\delta_s$ therefore gives
\[
 \mathbb E_{Q_s}\ell-\delta_s\le \mathbb E_{P_s}E
                  \le\mathbb E_{Q_s}u+\delta_s.
\]
For a bounded $v\in[0,1]$, the convention
$\|P-Q\|_{\rm TV}=\sup_A|P(A)-Q(A)|$ gives
$|\mathbb E_Pv-\mathbb E_Qv|\le\|P-Q\|_{\rm TV}$.
Replacing $Q_s$ by its empirical law costs at most $\varepsilon_n$ at either endpoint; replacing $P_s$ by $P_t$ costs at most $\rho$. Finally,
$|\mathbb E_tY-\mathbb E_tE|\le\mathbb P_t(Y\ne E)=\eta_t$.
Adding these inequalities and clipping to $[0,1]$ gives the main error envelope. It is deterministic conditional on the four error quantities. Any confidence interpretation requires estimates or probability bounds for those quantities.

For stationary scalar observations $V_i=\ell(O_i)$, with autocovariance $\gamma(h)=\operatorname{Cov}(V_i,V_{i+h})$,
\[
 \operatorname{Var}\!\left(\frac1n\sum_{i=1}^nV_i\right)
 =\frac{1}{n^2}\left[n\gamma(0)+2\sum_{h=1}^{n-1}(n-h)\gamma(h)\right].
\]
This follows by expanding the variance of the sum and grouping pairs by lag. The same formula applies to upper-endpoint observations. It explains the use of whole campaigns as resampling units instead of assigning independent Bernoulli variance to serially dependent attempts.

\subsection{Exhaustive finite checks}
The accompanying \texttt{scripts/verify\_analysis.py} provides an independent finite audit of the calculations in the main article. It is distinct from the application and timing data. A seven-coordinate model comprises two eligibility bits, two replica-demand bits, a deadline bit and two mandatory completion bits. Its 128 complete states are enumerated for every one of the 2,187 partial observations.

For all observations, the complete six-dimensional target image equals the image obtained from the two uniform completion endpoints per compatible control assignment. This checks all joint linear contrasts at once in that finite model, including the attainability of the retained witnesses. The closed-form decisive-state certificates also agree with full enumeration. All 10,206 one-bit observation refinements satisfy interval nesting. Separate examples check the incompatibility of independently optimized marginal endpoints and failure of uniform completion for a non-conjunction expression.

Exact rational checks exercise each of the four error-budget terms at an active boundary, followed by 250 combined cases with changes to the observation law, incorrect known bits, regime shift and semantic mismatch. Finally, the serial-covariance expression is compared with complete sequence enumeration for a stationary two-state Markov chain, over lengths one through nine. The chain stays in its current state with probability $4/5$; its nonzero serial covariance prevents an independent-request variance formula from passing the check.

All checks pass. The script writes \texttt{data/analytical-checks.json} and requires only the Python standard library. These finite audits can detect errors in the derived calculations; they do not establish ergodicity, complete instrumentation or semantic fidelity of the benchmark applications.

\section{Reproducibility sources}
The publication source is \href{https://github.com/a-a-k/telemetry-availability-identification/tree/61e65e4c09e25dd3c752756d07fe91ca504ece2f}{\texttt{simpat-artifacts-v2.1}}, commit \texttt{61e65e4c09e2}. Its \href{https://github.com/a-a-k/telemetry-availability-identification/blob/61e65e4c09e25dd3c752756d07fe91ca504ece2f/docs/EXPERIMENT_MAP.md}{experiment map} links English protocols, results and reproduction commands. The numerical component \path{Reproducibility_Zenodo.zip} supplies inputs and scripts; \path{docs/RESULTS.md} maps outputs to commands, and \path{docs/SOURCES.md} distinguishes current documents from original producer identities and hashes.

The package's \path{data/v4} retains confirmation forecasts, compatibility counts, missingness, solver measurements, pipeline pairs and source manifests. The generator \path{scripts/make_v4_artifacts.py} checks hashes, all 11,564 paired ratios, exact frequency equality and missingness count identities before producing tables and figures.

Confirmation acquisition uses source \texttt{12ddb09369ee}, repaired analysis and revised backends \texttt{a19fef0b998e}, and pipeline measurement \texttt{4d5270008e68}; manifests retain full identifiers. Publication cleanup preserves the evidence and mathematical implementation from scientific tree \texttt{e6b823017db0}. Earlier diagnosis and scaling inputs are under \path{data/historical}.

Native evidence comprises 310 confirmation/historical bundles (718,849,377 bytes), seven pipeline bundles (46,803,348 bytes), and eighteen completion-target bundles (Section~\ref{app:direct}). The enclosing upload ZIPs have different byte counts, recorded with SHA-256 in \path{Zenodo_payload_inventory.csv}; \path{evidence_manifests} identifies every inner bundle. Replay from raw observations requires these payloads and the frozen research code. The assigned evidence/software deposit identifier is \href{https://doi.org/10.5281/zenodo.22734273}{10.5281/zenodo.22734273}.

The \path{data/v5} inputs cover forecast pairs, 5,130 primitive-mask settings, 108 four-route processes and 11,160 ready queries. The generator \path{scripts/make_v5_artifacts.py} verifies twenty file hashes, exact forecast equality, support, masking and paired costs. The independent \path{scripts/verify_direct_image.py} checks all 2,187 partial observations of a seven-coordinate model and a nonuniform-law lift after category merging. Run \texttt{make reproduce} for local reconstruction and finite checks; native solvers and new acquisition are unnecessary for these computations.

\end{document}